\documentclass[11pt,letterpaper]{article}
\usepackage[utf8]{inputenc}
\usepackage{amsmath,amssymb,amsthm,mathrsfs,amsfonts}
\usepackage{bm}
\usepackage{color}
\usepackage{enumerate}
\usepackage{physics}
\usepackage{complexity}
\usepackage{soul,xcolor}
\usepackage{textcomp}
\usepackage{verbatim}
\usepackage{tikz}
\usetikzlibrary{arrows}
\usetikzlibrary{backgrounds}
\usetikzlibrary{quantikz}
\usepackage{caption}
\usepackage{subcaption}
\usepackage{algpseudocode}
\usepackage{algorithm}
\usepackage{makecell}
\usepackage[backend=biber,style=alphabetic,url=false,isbn=false,maxnames=10,minnames=3,maxalphanames=4,minalphanames=3]{biblatex}
\newtheorem{conjecture}{Conjecture}
\newtheorem{corollary}{Corollary}
\newtheorem{definition}{Definition}
\newtheorem{lemma}{Lemma}

\newtheorem{remark}{Remark}

\newtheorem{theorem}{Theorem}

\newcommand{\mc}{\mathcal}
\renewcommand{\E}{\mathop{\mathbb E\/}}

\newclass{\sharpp}{\#P}
\newclass{\cocequalp}{coC_{=}P}
\newfunc{\expp}{exp}

\newcommand{\xeb}{\mathrm{XEB}}
\newcommand{\xq}{\mathrm{XQ}}
\newcommand{\fsim}{\mathrm{fSim}}

\newcommand{\lket}[1]{\vert #1 \rangle\!\rangle}

\newcommand{\lbraket}[2]{\langle\!\langle #1 \vert #2 \rangle\!\rangle}

\newcommand{\lmel}[3]{\langle\!\langle #1 \vert #2 \vert #3 \rangle\!\rangle}

\usepackage[colorlinks = true]{hyperref}
\usepackage{longtable}

\usepackage{xcolor}
\definecolor{darkred}  {rgb}{0.5,0,0}
\definecolor{darkblue} {rgb}{0,0,0.5}
\definecolor{darkgreen}{rgb}{0,0.5,0}

\hypersetup{
  urlcolor   = blue,         
  linkcolor  = darkblue,     
  citecolor  = darkgreen,    
  filecolor   = darkred       
}

\usepackage[
letterpaper,
top=1in,
bottom=1in,
left=1in,
right=1in]{geometry}
\newtheorem*{theorem*}{Theorem}

\begin{document}
\title{A polynomial-time classical algorithm for noisy random circuit sampling}
\author{Dorit Aharonov\thanks{The Hebrew University. \href{mailto:dorit.aharonov@gmail.com}{dorit.aharonov@gmail.com}}
\and Xun Gao\thanks{Harvard University. \href{mailto:xungao@g.harvard.edu}{xungao@g.harvard.edu}}
\and Zeph Landau\thanks{UC Berkeley. \href{mailto:zeph.landau@gmail.com}{zeph.landau@gmail.com}}
\and Yunchao Liu\thanks{UC Berkeley. \href{mailto:yunchaoliu@berkeley.edu}{yunchaoliu@berkeley.edu}}
\and Umesh Vazirani\thanks{UC Berkeley. \href{mailto:vazirani@cs.berkeley.edu}{vazirani@cs.berkeley.edu}}}
\date{}
\maketitle

\begin{abstract}
We give a polynomial time classical algorithm for sampling from the output distribution of a noisy random quantum circuit in the regime of anti-concentration to within inverse polynomial total variation distance. This gives strong evidence that, in the presence of a constant rate of noise per gate, random circuit sampling (RCS) cannot be the basis of a scalable experimental violation of the extended Church-Turing thesis. Our algorithm is not practical in its current form, and does not address finite-size RCS based quantum supremacy experiments.
\end{abstract}

\section{Introduction}

Quantum random circuit sampling (RCS) is a basic primitive at the heart of recent ``quantum supremacy" experiments~\cite{Arute2019,Wu2021Strong,ZHU2022Quantum}. The quantum circuits in question are typically defined over a fixed architecture with gates chosen at random from some distribution (Fig.~\ref{fig:rcs}); in this work we assume two qubit gates which are Haar random\footnote{The requirement of Haar random 2-qubit gates can be relaxed; see Definition~\ref{def:architecture} and Remark~\ref{remark:assumption}.}. There are three parameters associated with the circuit: the number of qubits $n$, the circuit depth $d$, and the number of gates $m=\Theta(n d)$. In the experiments a relatively small number of samples are collected from the experimental implementation of RCS (though this number must necessarily scale exponentially in $d$), followed by a classical verification of these samples, using a statistical measure such as linear cross entropy (XEB), which requires classical post-processing time that is much larger and scales exponentially in $n$. Moreover in the experiments the depth $d$ is sufficiently large that the output distribution of the ideal random quantum circuit (Fig.~\ref{fig:rcs} (a)) is anti-concentrated\footnote{Anti-concentration is a property of random circuits which says that the output distribution is sufficiently flat when circuit depth is large enough. It was proven that anti-concentration holds for random circuits defined on 1D architecture~\cite{barak2021spoofing,Dalzell2022random} as long as circuit depth $d=\Omega(\log n)$, and it was conjectured that $\Omega(\log n)$ depth suffices for anti-concentration for any reasonably connected architecture such as 2D lattice~\cite{Dalzell2022random}, due to the fact that random circuits in 2D is expected to have faster mixing than 1D (Remark~\ref{remark:anticoncentration}).}, and indeed the output distribution tends to the Porter-Thomas distribution. 

Quantum supremacy is not only a milestone on the way to a practical quantum computer, it is also a fundamental physics experiment that tests quantum mechanics in the limit of high complexity --- it is ``an experimental violation of the extended Church-Turing thesis''~\cite{Bernstein1997Quantum,Aaronson2013Computational}. To demonstrate such a violation one must carry out a quantum computation that cannot be simulated by a polynomial time classical algorithm. This has been the main motivation for showing the complexity-theoretic hardness of both ideal and noisy random circuit sampling~\cite{Bouland2019on,Movassagh2019Quantum,Bouland2022Noise,Kondo2022Quantum,Krovi2022Average,Dalzell2021random,Aaronson2017Complexity,Aaronson2020on} (see Section~\ref{sec:priorwork} for a detailed discussion). However, recent work~\cite{Gao2021Limitations} cast doubt on the conjecture of~\cite{Aaronson2020on} that provided the hardness of noisy RCS based on the XEB test. Yet, this work left unclear whether hardness of the XEB test could be restored by formulating a new conjecture, and whether efficient classical algorithms exist for the other statistical tests for RCS output distributions such as the Heavy Output Generation (HOG) \cite{Aaronson2017Complexity} and log XEB~\cite{Boixo2018}.

In this paper we study the classical complexity of RCS in the presence of a constant rate of noise per gate. Specifically we consider a simple noise model shown in Fig.~\ref{fig:rcs} (b) where a (arbitrarily small) constant amount of depolarizing noise is applied to each qubit at each time step, which is a theoretical model for the actual RCS experiments. Our main result shows that sampling from the output distribution of a noisy random circuit can be approximately simulated by an efficient classical algorithm within small total variation distance. 

\begin{figure}[t]
    \centering
    \begin{subfigure}[b]{0.45\textwidth}
    \includegraphics[width=\linewidth]{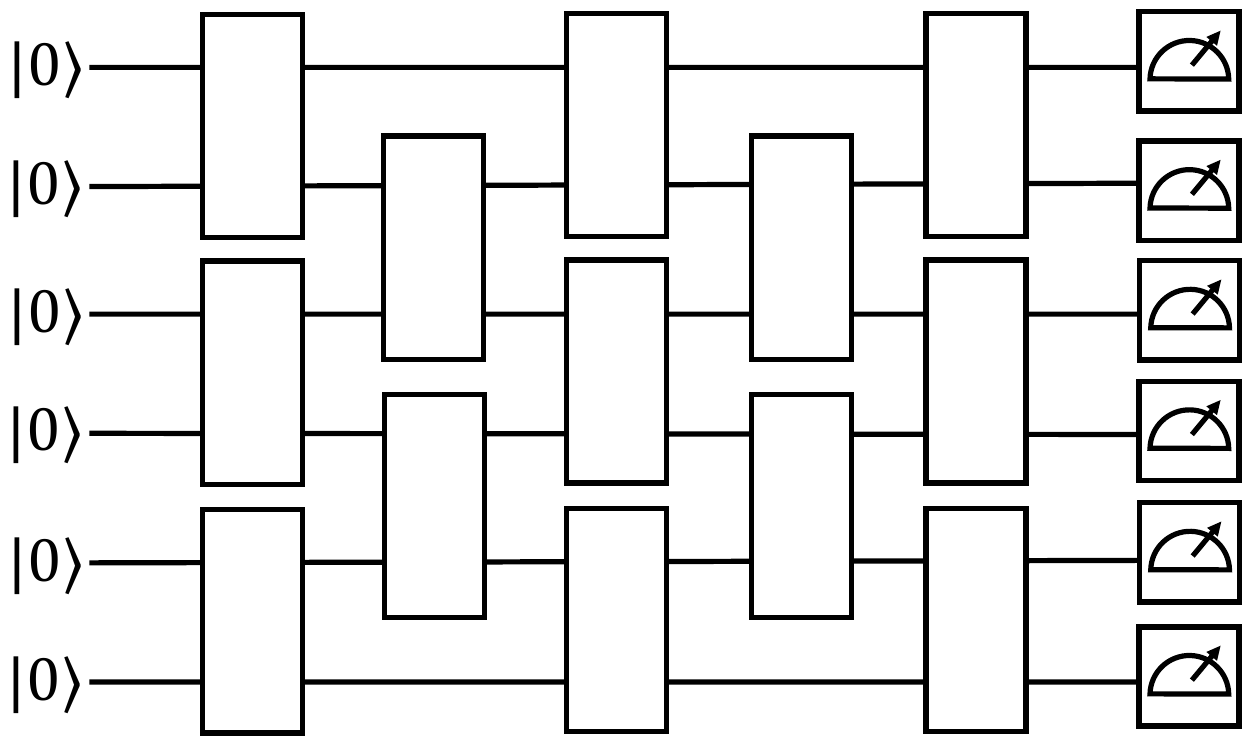}
    \caption{Ideal RCS}
    \end{subfigure}
    \hfill
    \begin{subfigure}[b]{0.45\textwidth}
    \includegraphics[width=\linewidth]{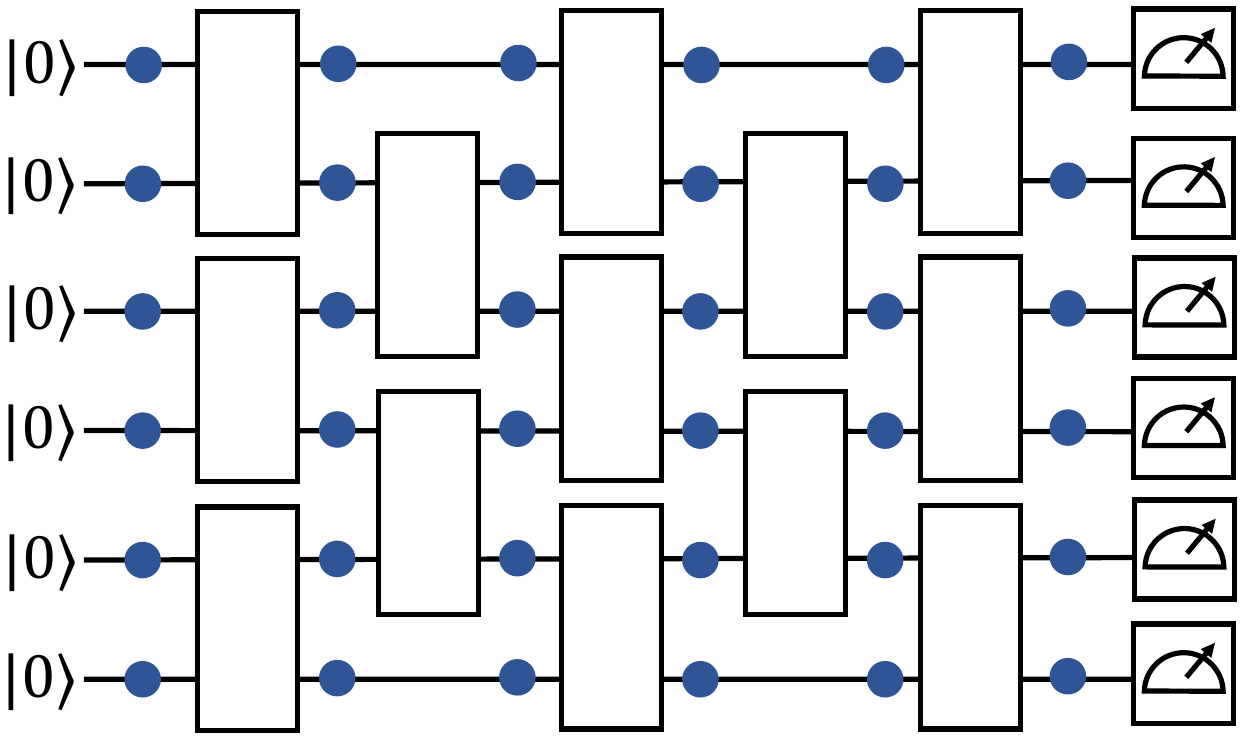}
    \caption{Noisy RCS}
    \end{subfigure}
    \caption{Random circuit sampling, each white box is an independent Haar random 2-qubit gate. (a) Ideal RCS generates an output distribution $p(C)$ that satisfies anti-concentration when $d=\Omega(\log n)$. (b) Noisy RCS, where an arbitrarily small constant amount of depolarizing noise is applied to each qubit at each step, which generates a noisy output distribution $\tilde{p}(C)$. Here the 1D architecture is for illustration; the result applies to general architectures (Definition~\ref{def:architecture}).}
    \label{fig:rcs}
\end{figure}

\begin{theorem}[Main result]
\label{thm:mainresult}
Assuming anti-concentration, there is a classical algorithm that, on input a random circuit $C$ on any fixed architecture, outputs a sample from a distribution that is $\varepsilon$-close to the noisy output distribution $\tilde{p}(C)$ in total variation distance with success probability at least $1-\delta$ over the choice of $C$, in time $\poly(n,1/\varepsilon,1/\delta)$.
\end{theorem}

To put this in perspective, consider a RCS quantum supremacy experiment that collects $M$ samples. We claim that Theorem~\ref{thm:mainresult} implies that there is a classical algorithm running in time bounded by polynomial in $M$, that outputs $M$ samples that are indistinguishable, i.e.~no statistical test can distinguish the output of the algorithm from the output of the experiment with probability greater than $1/2 + \mu$, for any constant $\mu>0$. This is because to achieve statistical indistinguishability it suffices to choose $\varepsilon=\mu/M$, which by the main result above gives a running time $\poly(n,M/\mu)$. Thus the running time of our algorithm is at most a polynomial in the running time of the experiment.

\begin{corollary}\label{cor:indistinguishability}
Assuming anti-concentration, no statistical test applied to $M$ samples can distinguish between the output of a noisy random circuit and the above classical algorithm with running time $\poly(n,M)$. In particular, if $M = \poly(n)$, the classical algorithm runs in $\poly(n)$ time.
\end{corollary}

We note that the implications of our result are complexity theoretic and do not directly address the soundness of finite-size quantum supremacy experiments. 

Also note that anti-concentration is a central assumption for both the RCS experiments and our algorithm, which is believed to hold for general architectures above $\Omega(\log n)$ depth~\cite{Dalzell2022random}. At the same time, the output distribution of noisy random circuits is $2^{-\Theta(d)}$ close to uniform in total variation distance~\cite{Aharonov1996limitations,Gao2018efficient,Deshpande2021tight}. 
This means that any quantum supremacy experiment must collect $M = 2^{\Omega(d)}$ samples. Thus $d = \Theta(\log n)$ was recognized as the sweet spot for scalable experimental demonstration of quantum computational advantage~\cite{Deshpande2021tight}, depth $O(\log n)$ to guarantee polynomial number of samples and $\Omega(\log n)$ to guarantee anti-concentration. In this regime both the sample complexity of the experiment and running time of our classical algorithm scale polynomially in $n$.

Our approach builds upon the work of Gao and Duan in 2018~\cite{Gao2018efficient}. They developed the idea of performing a Fourier transform on quantum circuits and an algorithm for simulating noisy random circuits via a truncation in Fourier domain and calculating low-degree Fourier coefficients. They used the resulting algorithm to efficiently estimate local observables for random analogs of fault-tolerance circuits, thus showing that structure is necessary for quantum fault-tolerance. While not explicitly mentioned in \cite{Gao2018efficient}, their approach in fact produces a quasi-polynomial time algorithm for sampling from the output distribution to within inverse polynomial total variation distance\footnote{This fact was unknown at the time, as it was believed that anti-concentration requires large circuit depth. Recent developments \cite{barak2021spoofing,Dalzell2022random} suggest otherwise.}. This raises the challenge of giving a polynomial time algorithm for the sampling problem.

We start by reformulating the Fourier transform defined by~\cite{Gao2018efficient} as Feynman path integral in the Pauli basis, and the simulation algorithm as calculating those Feynman paths with lowest Hamming weight. The Pauli basis framework was also used by \cite{Gao2021Limitations} to give an alternative argument for achieving a $2^{-O(d)}$ XEB (see Appendix~\ref{sec:xquath} for a more formal treatment). The advantage of using the Pauli basis for Feynman path integral is that most low-Hamming-weight Feynman paths have 0 contribution to the path integral. This view helps design an enumeration algorithm that calculates the contributions of only non-trivial paths in polynomial time. From the perspective of Fourier analysis, prior algorithms of \cite{Bremner2017achievingquantum,Gao2018efficient} based on low-degree Fourier approximation mainly rely on noise sensitivity and have running time $n^{O(\log 1/\varepsilon)}$ where $\varepsilon$ is the desired approximation error which results in quasi-polynomial running time for our purpose, but our algorithm has running time $2^{O(\log 1/\varepsilon)}=\poly(1/\varepsilon)$ due to the additional property of Fourier sparsity. 

Our algorithm is not practical in its current form due to a large exponent in the running time, and we leave as an interesting future direction to develop practical implementations using our framework. See Section~\ref{sec:remarks} for discussions regarding finite-size noisy RCS experiments.

\subsection{Description of algorithm}
\label{sec:descriptionofalgorithm}

Let $\rho$ be an $n$ qubit density matrix. We can write $\rho=\sum_{s\in\mathsf{P}_n}\alpha_s\cdot s$ where $\mathsf{P}_n$ are the normalized $n$-qubit Pauli operators, and $\alpha_s = \Tr(s\rho)$ is real. We keep track of the coefficients in the Pauli basis after unitary evolution $\rho\mapsto U\rho U^\dag$, which evolve according to the rule $\Tr(s U\rho U^\dag)=\sum_{t\in\mathsf{P}_n} \Tr(s U t U^\dag) \Tr(t\rho)$. Comparing with the transition rule $\mel{x}{U}{\psi}=\sum_{y}\mel{x}{U}{y}\braket{y}{\psi}$ we can see that while $\mel{x}{U}{y}$ is the transition amplitude from $\ket{y}$ to $\ket{x}$, $\Tr(s U t U^\dag)$ plays the role of transition amplitude from $t$ to $s$.

Consider a quantum circuit $C=U_d U_{d-1}\cdots U_1$ where $U_i$ is a layer of 2-qubit gates and $d$ is circuit depth. A \textbf{Pauli path} is a sequence $s=(s_0,\dots,s_d)\in\mathsf{P}_n^{d+1}$. The Feynman path integral in the Pauli basis (in short, Pauli path integral) is written as sum of product of transition amplitudes,
\begin{equation}
    p(C,x)=\sum_{s_0,\dots,s_d\in\mathsf{P}_n}\Tr(\ketbra{x}s_d)\Tr(s_d U_d s_{d-1} U_d^\dag)\cdots\Tr(s_1 U_1 s_0 U_1^\dag)\Tr(s_0\ketbra{0^n}).
\end{equation}
Note that LHS is the probability $p(C,x)=\left|\mel{x}{C}{0^n}\right|^2$ instead of amplitude. Denote the contribution of a Pauli path $s=(s_0,\dots,s_d)\in\mathsf{P}_n^{d+1}$ to the path integral as $f(C,s,x)$, which gives $p(C,x)=\sum_{s\in\mathsf{P}_n^{d+1}}f(C,s,x)$.

Our algorithm for simulating noisy random circuits is based on a simple but powerful fact, used in \cite{Kempe2010Upper,Gao2018efficient}. Consider the single-qubit depolarizing noise with strength $\gamma$, $\mc E(\rho):=(1-\gamma)\rho + \gamma\frac{I}{2}\Tr(\rho)$. Then the contribution of a Pauli path of a noisy quantum circuit subject to this noise, decays exponentially with the Hamming weight of the Pauli path:
\begin{equation}
    \tilde{p}(C,x)=\sum_{s\in\mathsf{P}_n^{d+1}}(1-\gamma)^{|s|}f(C,s,x),
\end{equation}
where $\tilde{p}(C,x)$ is the output probability of the noisy circuit and $|s|$ is the \textbf{Hamming weight} of $s$ (the number of non-identity Pauli in $s$).
We would like to approximate the value $\tilde{p}(C,x)$ by summing only over the low-weight Pauli paths,
\begin{equation}\label{eq:introapprox}
    \tilde{p}(C,x)\approx\sum_{s\in\mathsf{P}_n^{d+1}:|s|\leq \ell}(1-\gamma)^{|s|}f(C,s,x),
\end{equation}
and claim that the total variation distance achieved by the approximation is $2^{-\Omega(\ell)}$ on average. This is not immediate since the $f(C,s,x)$ can be both positive and negative. We invoke two properties of random circuits: the first is \textbf{orthogonality}, which says that on average over random circuits the product of the contributions from two different Pauli paths equals 0, i.e.~$\E_C[f(C,s,x)f(C,s',x)]=0$ when $s\neq s'$; the second is \textbf{anti-concentration}, which says that the sum of squares of the output probability of a random circuit is small, i.e.~$\E_C\sum_x p(C,x)^2=O(1)\cdot 2^{-n}$. Roughly speaking, orthogonality allows us to upper bound the total variation distance by a sum of squares quantity, which is then upper bounded using anti-concentration.

The next step is to develop an algorithm to calculate the RHS of Eq.~\eqref{eq:introapprox}. Note that a straightforward sum over all paths up to weight $\ell$ gives a running time of $O(n d)^{O(\ell)}$ leading to a quasi-polynomial time algorithm as in~\cite{Gao2018efficient}. Here we develop a \textbf{counting argument} and efficient enumeration method for all Pauli paths of weight at most $\ell$ which takes only $2^{O(\ell)}$ time. The idea is \textbf{sparsity} of the low-weight paths, meaning that for most Pauli paths in $\mathsf{P}_n^{d+1}$, its contribution $f(C,s,x)$ is $0$; therefore we design a combinatorial algorithm that only enumerates those paths that have non-zero contributions. Finally, the sampling algorithm follows from a general sampling-to-computing reduction of \cite{Bremner2017achievingquantum}.

At a high level, the hardness assumptions in~\cite{Aaronson2017Complexity,Aaronson2020on} may be intuitively viewed as asserting that Feynman path integral in the computational basis is essentially the best classical algorithm for RCS, and achieving non-trivial correlation requires following exponentially many paths. Instead, the Pauli path integral approach has the virtue that low weight paths have the most significant contribution.

\subsection{Prior work regarding the computational complexity of RCS}
\label{sec:priorwork}

To put the above results in context, let us recall the background regarding complexity theoretic evidence that classical computers cannot efficiently sample from the output of a random quantum circuit (this section focuses on asymptotic hardness; see next section for discussions regarding finite-size experiments). There are two main genres of results along those lines, which we review below (see \cite{Hangleiter2022Computational} for a more comprehensive survey).

The first is in the form of a worst-case to average-case reduction, showing that if an efficient classical algorithm can sample from the output distribution of ideal RCS within small total variation distance, then the Polynomial Hierarchy collapses~\cite{Bouland2019on,Movassagh2019Quantum,Bouland2022Noise,Kondo2022Quantum,Krovi2022Average}. 
The eventual goal of this program was to show classical hardness for sampling within constant total variation distance, which would require showing average-case hardness of computing the output probability of ideal RCS within additive error $O(2^{-n})$. While the earliest average-case hardness results could only tolerate very small additive error, it was hoped that over time the reductions could be made more robust. This has indeed been the case, with an improvement from a large polynomial in the exponent~\cite{Bouland2019on} to $2^{-O(m)}$~\cite{Krovi2022Average}, but this line of work has hit an obstacle that may prove difficult to overcome (see e.g.~\cite[Section 3]{Bouland2022Noise} and \cite[Section II A]{Deshpande2021tight}). Moreover, these results do not address the actual RCS experiments which are highly noisy\footnote{The related work of \cite{Dalzell2021random} argued that the hardness of approximate sampling for noisy RCS can be reduced to ideal RCS, but the argument required a local noise model that decreases as $\tilde{O}(1/n)$, which is not scalable.}.

The second genre is based on complexity theoretic assumptions about the difficulty of distinguishing heavy and light outputs of the random circuit~\cite{Aaronson2017Complexity,Aaronson2020on}. These assumptions essentially say that even a tiny correlation (order $2^{-n}$) with the output distribution of ideal RCS is hard to achieve classically. While these assumptions are quite strong, they have the virtue of yielding robust bounds. Indeed a specific conjecture in this genre called XQUATH~\cite{Aaronson2020on} has provided robust complexity theoretic foundation of the linear cross entropy benchmark (XEB) used in recent experiments~\cite{Arute2019,Wu2021Strong,ZHU2022Quantum}. This provided a way to heuristically argue that even the very small XEB achieved in actual 50-70 qubit experiments was a classically difficult computational task. However, the strong parameters in the assumption (correlation of order $2^{-n}$) was called into question by the result of \cite{Gao2021Limitations}, although it remained unclear if the hardness of the XEB test can be restored by changing the parameters in XQUATH. In addition, it was unclear whether the hardness of the other statistical tests such as HOG or log XEB was impacted. Our results address these questions by showing that no statistical tests, like the XEB, HOG and log XEB, can distinguish between noisy RCS and our classical algorithm.

\subsection{Concluding remarks: what our results do not address}
\label{sec:remarks}
We importantly note several points left unaddressed by our results. 
\begin{itemize}
    \item {\bf Practical speed-ups.}
We note that our results do not address RCS based quantum supremacy in its non-asymptotic, practical form. In particular, much progress has been made in developing practical spoofing algorithms for achieving a similar numerical value as the XEB in current 53-60 qubit RCS experiments. Practical tensor network algorithms \cite{Gray2021hyperoptimized,Huang2020Classical,Pan2022Simulation,Kalachev2021Multi,Pan2022Solving,Kalachev2021Classical} can achieve this goal using hundreds of GPUs in a few hours, but these algorithms have exponential scaling and become impractical if the system size increases by a few qubits. A numerical implementation of the algorithm in~\cite{Gao2021Limitations} achieved roughly 10\% of Google's XEB using 1 GPU in 1 second, though it remains unclear whether this algorithm can achieve Google's XEB (using much less than hundreds of GPUs).
Our algorithm is not practical in its current form, as there is a large constant (of order $1/\gamma$ where $\gamma$ is the error per gate) in the degree of the polynomial of the running time. An interesting future direction is to develop practical implementations using our framework and ideas from \cite{Gao2021Limitations} that achieves similar XEB as in the experiments~\cite{Arute2019,Wu2021Strong,ZHU2022Quantum} using a small amount of resource.

\item {\bf Sublogarithmic depth.} 
Our algorithm assumes anti-concentration and therefore works for random circuits with depth at least $\Omega(\log n)$.\footnote{It was shown \cite{Dalzell2022random,Deshpande2021tight} that anti-concentration requires at least $\Omega(\log n)$ depth for random circuits with Haar random 2-qubit gates.} The issue with sub-logarithmic depth random circuits (with Haar random 2-qubit gates) is that there is no evidence for hardness of sampling even for ideal RCS, as all existing results for average-case hardness (the first genre discussed above) are only relevant for sampling when anti-concentration holds. In addition, \cite{Napp2022Efficient} gives evidence that 2D ideal RCS can be efficiently simulated when depth is smaller than some fixed constant. The complexity of ideal and noisy RCS remains unclear at depth between constant and $o(\log n)$. Separately, existing quantum supremacy experiments rely on the assumption that the ideal circuit is close to Porter-Thomas for benchmarking; closeness to Porter-Thomas is even stronger than anti-concentration.  

Notwithstanding the above discussion, it remains possible that a different approach based on RCS of sublogarithmic depth circuits, which does not rely on anti-concentration, could lead to a scalable experimental violation of the extended Church-Turing thesis.

\item{\bf Less random gate sets.}
Besides anti-concentration, our algorithm also requires randomness in the gate set.
The simplest distribution over the gate set to think of is that of Haar random 2-qubit gates. However, the gate sets used in actual experiments~\cite{Arute2019,Wu2021Strong,ZHU2022Quantum} are not Haar random 2-qubit gates, but gates with more limited 
randomness. While we do not know if our results hold for the exact gate sets used in those recent experiments, we show in Section~\ref{sec:gateset} that our algorithm works for a gate set which is closely related to the gate sets used in those experiments; more generally, the required condition for our results is in fact much weaker than Haar random 2-qubit gates (see Definition~\ref{def:architecture}).
\end{itemize}

\noindent{\bf Overview of remainder of paper.} In Section~\ref{sec:paulibasis} we give formal definitions of the Pauli path integral and derive useful properties of this framework. In Section~\ref{sec:mainresult} we give the proof of our main result, and discuss Google and USTC's gate set in Section~\ref{sec:gateset}. Appendix~\ref{sec:xquath} contains a formal proof for refuting XQUATH using the Pauli basis framework. As an application of the Pauli basis framework, we provide simple proofs for existing results about random circuits, including a lower bound on the depth for anti-concentration previously shown by~\cite{Dalzell2022random} (Corollary~\ref{cor:aclowerbound}), and an improved lower bound on the convergence to uniform for noisy random circuits previously shown by~\cite{Deshpande2021tight} (Appendix~\ref{sec:convergetouniform}).

\section{The Pauli basis framework}
\label{sec:paulibasis}
We first give formal definitions of the Pauli path integral discussed in Section~\ref{sec:descriptionofalgorithm} and then derive useful properties of this framework.

Let $C=U_d U_{d-1}\cdots U_1$ be a quantum circuit acting on $n$ qubits, where $U_i$ is a layer of 2-qubit gates and $d$ is circuit depth. The Feynman path integral in the computational basis is written as
\begin{equation}\label{eq:feynmanpathcomputational}
    \expval{C}{0^n}=\sum_{x_1,\dots,x_{d-1}\in\{0,1\}^n}\mel{0^n}{U_d}{x_{d-1}}\mel{x_{d-1}}{U_{d-1}}{x_{d-2}}\cdots \mel{x_1}{U_1}{0^n}.
\end{equation}

The main difference when switching to the Pauli basis is that instead of thinking about a quantum circuit as applying unitary matrices to vectors, we think of it as unitary channels applied to density matrices, $\mc C=\mc U_d \mc U_{d-1}\cdots \mc U_1$ where each $\mc U_i(\cdot):=U_i(\cdot) U_i^\dag$ is a unitary channel. Similar to decomposing a pure state vector into a superposition of computational basis states, we consider the normalized Pauli operators
\begin{equation}
    \mathsf{P}_n:=\left\{I/\sqrt{2},X/\sqrt{2},Y/\sqrt{2},Z/\sqrt{2}\right\}^{\otimes n}
\end{equation}
as an operator basis and decompose a density matrix into a linear combination of Pauli operators (Table~\ref{tab:paulibasis}). In Table~\ref{tab:paulibasis} we present the operator basis as a direct analogy of vector basis by switching to the operator ket notation (Table~\ref{tab:paulibasis} (c)).

\begin{table}[t]
    \centering
    \begin{tabular}{|l|l|l|l|}
    \hline
         & (a) Vector basis & (b) Operator basis & (c) Operator basis  \\\hline
       State & $\displaystyle\ket{\psi}=\sum_{x\in\{0,1\}^n}\braket{x}{\psi}\ket{x}$ & $\displaystyle\rho=\sum_{s\in\mathsf{P}_n}\Tr(s\rho) s$ & $\displaystyle\lket{\rho}=\sum_{s\in\mathsf{P}_n}\lbraket{s}{\rho}\lket{s}$ \\\hline
       Evolution & $\displaystyle\ket{\psi}\mapsto U\ket{\psi}$ & $\displaystyle\rho\mapsto U\rho U^\dag$ & $\displaystyle\lket{\rho}\mapsto \mc U \lket{\rho}$ \\\hline
       Path integral & \parbox[l]{3.9cm}{\begin{equation*}
           \begin{aligned}
          &\mel{x}{U}{\psi}\\=&\sum_{y\in\{0,1\}^n}\mel{x}{U}{y}\braket{y}{\psi}
           \end{aligned}
       \end{equation*}} 
       & \parbox[l]{4.2cm}{\begin{equation*}
           \begin{aligned}
               &\Tr(s U\rho U^\dag)\\=&\sum_{t\in\mathsf{P}_n} \Tr(s U t U^\dag) \Tr(t\rho)
           \end{aligned}
       \end{equation*}} 
       & \parbox[l]{3.3cm}{\begin{equation*}
           \begin{aligned}
            &\lmel{s}{\mc U}{\rho}\\=&\sum_{t\in\mathsf{P}_n} \lmel{s}{\mc U}{t}\lbraket{t}{\rho} 
           \end{aligned}
       \end{equation*}}\\\hline
    \end{tabular}
    \caption{The Feynman path integral can be viewed as decomposing the state into basis states at each step of time evolution. (a) The standard decomposition with computational basis states. (b) Decomposition using the Pauli operator basis, where states are represented as density matrices and time evolution is represented as unitary channels. (c) The same decomposition using the Pauli operator basis, presented with operator ket notation.}
    \label{tab:paulibasis}
\end{table}

\begin{definition}[Pauli path integral]
Let $C=U_d U_{d-1}\cdots U_1$ be a quantum circuit acting on $n$ qubits, where $U_i$ is a layer of 2-qubit gates and $d$ is circuit depth, and let $p(C,x):=\left|\mel{x}{C}{0^n}\right|^2$ be the output probability distribution. The Pauli path integral is written as
\begin{equation}\label{eq:feynmanpathpauli}
\begin{aligned}
    p(C,x)&=\sum_{s_0,\dots,s_d\in\mathsf{P}_n}\Tr(\ketbra{x}s_d)\Tr(s_d U_d s_{d-1} U_d^\dag)\cdots\Tr(s_1 U_1 s_0 U_1^\dag)\Tr(s_0\ketbra{0^n})\\
    &=\sum_{s_0,\dots,s_d\in\mathsf{P}_n}\lbraket{x}{s_d}\lmel{s_d}{\mc U_d}{s_{d-1}}\cdots \lmel{s_1}{\mc U_1}{s_0}\lbraket{s_0}{0^n}.
\end{aligned}
\end{equation}
Here each term on RHS corresponds to a \textbf{Pauli path} $s=(s_0,\dots,s_d)\in\mathsf{P}_n^{d+1}$. We also define the \textbf{Fourier coefficient} of a quantum circuit $C$ with output $x$ and Pauli path $s$ as
\begin{equation}
    f(C,s,x):=\lbraket{x}{s_d}\lmel{s_d}{\mc U_d}{s_{d-1}}\cdots \lmel{s_1}{\mc U_1}{s_0}\lbraket{s_0}{0^n}
\end{equation}
and the output probability is written as
\begin{equation}
    p(C,x)=\sum_{s\in\mathsf{P}_n^{d+1}}f(C,s,x).
\end{equation}
\end{definition}

Eq.~\eqref{eq:feynmanpathpauli} follows from repeatedly applying the rules shown in Table~\ref{tab:paulibasis}. The above definition can also be extended to noisy quantum circuits. Let $\mc E(\rho):=(1-\gamma)\rho + \gamma\frac{I}{2}\Tr(\rho)$ be the single-qubit depolarizing noise with strength $\gamma$. It has the property that $\mc E(I)=I$ and $\mc E(P)=(1-\gamma)P$ when $P\in\left\{X,Y,Z\right\}$. 

\begin{definition}[Pauli path integral for noisy quantum circuits]\label{def:noisyfeynmanpath}
For a quantum circuit $C=U_d U_{d-1}\cdots U_1$, let $\tilde{C}$ be a noisy quantum circuit where each qubit in $C$ is subject to $\gamma$ depolarizing noise in each layer (Fig.~\ref{fig:rcs} (b)). Let
\begin{equation}
    \tilde{p}(C,x):=\lmel{x}{\mc E^{\otimes n}\mc U_d\mc E^{\otimes n}\cdots \mc U_1\mc E^{\otimes n}}{0^n}
\end{equation}
be the output probability distribution of the noisy circuit $\tilde{C}$. The Pauli path integral for $\tilde{C}$ is defined as
\begin{equation}
    \tilde{p}(C,x)=\sum_{s\in\mathsf{P}_n^{d+1}}\tilde{f}(C,s,x)
\end{equation}
where
\begin{equation}
    \tilde{f}(C,s,x):=\lmel{x}{\mc E^{\otimes n}}{s_d}\lmel{s_d}{\mc U_d\mc E^{\otimes n}}{s_{d-1}}\cdots \lmel{s_1}{\mc U_1\mc E^{\otimes n}}{s_0}\lbraket{s_0}{0^n}.
\end{equation}
Let $|s|$ be the \textbf{Hamming weight} of $s$ (the number of non-identity Pauli in $s$). The definition of depolarizing noise implies that
\begin{equation}\label{eq:noise_decay}
    \tilde{f}(C,s,x)=(1-\gamma)^{|s|}f(C,s,x).
\end{equation}
\end{definition}

\begin{figure}[t]
  \begin{algorithm}[H]
    \caption{Simulating noisy random circuits by low-degree Fourier approximation}\label{alg:mainsimulation}
    \raggedright\textbf{Input:} quantum circuit $C$, truncation parameter $\ell$, $x\in\{0,1\}^n$\\
    \textbf{Output:} an approximation of $\tilde{p}(C,x)$
    \begin{algorithmic}[1]
    \State $q\leftarrow 0$
    \ForAll{legal Pauli path $s$ with $|s|\leq \ell$}
    \State calculate $f(C,s,x)$
    \State $q\leftarrow q+(1-\gamma)^{|s|}f(C,s,x)$
    \EndFor
    \State\textbf{Return} $q$
    \end{algorithmic}
    \end{algorithm}
\end{figure}

Our algorithm described in Section~\ref{sec:descriptionofalgorithm} is summarized in Algorithm~\ref{alg:mainsimulation} (``legal'' Pauli path is defined in Definition~\ref{def:legalpath}). Next we develop properties of the Pauli basis that are useful later.

First, note that the Fourier coefficients $f(C,s,x)$ can be further decomposed into products of transition amplitudes of 2-qubit gates $\lmel{q}{\mc U}{p}=\Tr(q U p U^\dag)$ where $U\in\mathbb{U}(4)$, $p,q\in\mathsf{P}_2$, so any Fourier coefficient can be computed in time $O(n d)$. The Fourier coefficients satisfy $f(C,s,x)\in\mathbb{R}$ and $\left|f(C,s,x)\right|\leq\frac{1}{2^n}$. This is because for any $x\in\{0,1\}^n$ and $s\in\mathsf{P}_n$ we have
\begin{equation}\label{eq:boundary}
    \lbraket{x}{s}=\Tr(\ketbra{x}s)\in\left\{0,-\frac{1}{\sqrt{2^n}},\frac{1}{\sqrt{2^n}}\right\}.
\end{equation}
In addition, the output $x$ only affects the sign of the Fourier coefficient, as
\begin{equation}\label{eq:boundarysign}
    f(C,s,x)^2=f(C,s,0^n)^2,\quad \forall x\in\{0,1\}^n.
\end{equation}

The rest of the properties we develop in this section crucially rely on the randomness of the gate set. We first recall the properties of Haar random 2-qubit gates.

\begin{lemma}[Properties of Haar random 2-qubit gates~\cite{Harrow2009}]\label{lemma:haarrandomgate}
Let $U\in\mathbb{U}(4)$ be a Haar random 2-qubit gate, and $p,q,r,s\in\mathsf{P}_2$. Then
    \begin{equation}\label{eq:haarorthogonality}
        \E_{U\sim\mathbb{U}(4)}\lmel{p}{\mc U}{q}\lmel{r}{\mc U}{s}=0\,\,\,\,\text{if }p\neq r\text{ or }q\neq s.
    \end{equation}
    We also have
    \begin{equation}\label{eq:2qubitgate}
        \E_{U\sim\mathbb{U}(4)}\lmel{p}{\mc U}{q}^2=\begin{cases}
            1, & p=q=I^{\otimes 2}/2,\\
            0, & p=I^{\otimes 2}/2, q\neq I^{\otimes 2}/2,\\
            0, & p\neq I^{\otimes 2}/2, q= I^{\otimes 2}/2,\\
            \frac{1}{15}, & \text{else}.
        \end{cases}
    \end{equation}
\end{lemma}

Eq.~\eqref{eq:haarorthogonality} is a key property which we refer to as gate-set orthogonality. It says that if we consider the Pauli basis decomposition and average over two copies of a random unitary, then the randomness forces the input and output Paulis to be the same across the two copies. Next we show that this property does not require full randomness over $\mathbb{U}(4)$; randomness over Pauli operators already suffices.

\begin{lemma}[Gate-set orthogonality]\label{lemma:gatesetorthogonality}
Let $\mc D$ be any distribution over $\mathbb{U}(4)$ that is invariant under right-multiplication of random Pauli, i.e. for any measurable function $F$,
\begin{equation}
    \E_{U\sim\mc D}[F(U)]=\E_{U\sim\mc D}\E_{V\sim\{I,X,Y,Z\}^2}[F(UV)].
\end{equation}
Then for any $P,Q\in \{I,X,Y,Z\}^2$ such that $P\neq Q$, we have
\begin{equation}
    \E_{U\sim\mc D}\left[U P U^\dag \otimes U Q U^\dag\right]=0.
\end{equation}
\end{lemma}
\begin{proof}
    Due to invariance under right-multiplication of random Pauli and linearity, it suffices to prove that 
    \begin{equation}
        \E_{V\sim\{I,X,Y,Z\}^2}\left[V P V^\dag \otimes V Q V^\dag\right]=0 \quad \text{if }P\neq Q.
    \end{equation}
    Let $\langle P,Q \rangle:=1[P,Q\text{ anticommute}]$. Then
    \begin{equation}
    \begin{aligned}
        \E_{V\sim\{I,X,Y,Z\}^2}\left[V P V^\dag \otimes V Q V^\dag\right]&=\frac{1}{16}\sum_{V\in\{I,X,Y,Z\}^2}V P V^\dag \otimes V Q V^\dag\\
        &=\frac{1}{16}\sum_{V\in\{I,X,Y,Z\}^2}(-1)^{\langle V,P\rangle+\langle V,Q\rangle}P \otimes Q\\
        &=\frac{1}{16}\sum_{V\in\{I,X,Y,Z\}^2}(-1)^{\langle V,P Q\rangle}P \otimes Q\\
        &=0,
    \end{aligned}
    \end{equation}
where the last line follows from the fact that $PQ$ is not identity, and therefore commutes with half Paulis and anticommutes with the other half.
\end{proof}

Our main result holds for any gate set and architecture that satisfies gate-set orthogonality and anti-concentration. We discuss these two properties separately and start with orthogonality.

\begin{definition}[Gate set and architecture of random circuits]\label{def:architecture}
We consider random quantum circuits defined over a fixed architecture described as follows. In each layer, each qubit experiences a 2-qubit gate (so the number of qubits $n$ is even, and there are $n/2$ 2-qubit gates per layer). The 2-qubit gates can be applied to any pair of qubits, without geometric locality. Each 2-qubit gate is independently drawn from some distribution that is invariant under right-multiplication of random Pauli. The final layer is drawn from a distribution that is invariant under both left- and right-multiplication of random Pauli.
\end{definition}

Note that the requirement that each qubit experiences a 2-qubit gate in each layer is for convenience; more general architectures can be handled by a suitable redefinition of circuit depth (this was also noted in \cite{Deshpande2021tight}).

Examples of gate sets that satisfy Definition~\ref{def:architecture} include Haar random 2-qubit gates as well as a fixed 2-qubit gate surrounded by Haar random single qubit gates. A fixed 2-qubit gate surrounded by random Pauli gates also satisfies Definition~\ref{def:architecture} but may violate anti-concentration (see Remark~\ref{remark:randompauli}). Any ensemble of random circuits that satisfies Definition~\ref{def:architecture} has the following crucial property that we frequently use.

\begin{lemma}[Orthogonality of Fourier coefficients]\label{lemma:orthogonality}
Let $C$ be a random circuit drawn from some distribution $\mc D$ that satisfies Definition~\ref{def:architecture}. Then for any Pauli paths $s\neq s'\in\mathsf{P}_n^{d+1}$ and for any $x\in\{0,1\}^n$ we have
\begin{equation}\label{eq:orthogonality}
    \E_{C\sim \mc D}\left[f(C,s,x)f(C,s',x)\right]=0.
\end{equation}
\end{lemma}
\begin{proof}
    As $s\neq s'$, there exists a 2-qubit gate $U$ that contributes transition amplitude $\lmel{q_1}{\mc U}{p_1}$ to $f(C,s,x)$ and contributes $\lmel{q_2}{\mc U}{p_2}$ to $f(C,s',x)$, such that $p_1\neq p_2\in\mathsf{P}_2$. Lemma~\ref{lemma:gatesetorthogonality} implies that
    \begin{equation}
        \E_{U}\left[\lmel{q_1}{\mc U}{p_1}\lmel{q_2}{\mc U}{p_2}\right]=0.
    \end{equation}
    Due to the independence between different gates, we can separately calculate the expectation over each gate in Eq.~\eqref{eq:orthogonality}. Therefore the above equation implies that the overall expectation in Eq.~\eqref{eq:orthogonality} equals 0. One special case is that the difference between $s$ and $s'$ happens at the last step $s_{d}$. For this case we use the left-invariance under random Pauli of the final layer of gates.
\end{proof}

Next we discuss anti-concentration, which is formally defined as follows.

\begin{definition}[Anti-concentration]\label{def:anticoncentration}
A distribution over quantum circuits $\mc D$ satisfies anti-concentration if
\begin{equation}
    \E_{C\sim\mc D}2^n\sum_{x\in\{0,1\}^n}p(C,x)^2=O(1).
\end{equation}
\end{definition}
\begin{remark}\label{remark:anticoncentration}
The following is known about anti-concentration:
\begin{itemize}
    \item \cite{barak2021spoofing,Dalzell2022random} showed that anti-concentration is satisfied for 1D random circuits with Haar random 2-qubit gates as long as circuit depth is above some constant times $\log n$.
    \item \cite{Dalzell2022random} also showed that $\Theta(n\log n)$ 2-qubit gates are necessary and sufficient for anti-concentration for a stochastic all-to-all connected architecture with Haar random 2-qubit gates.
    \item \cite{Dalzell2022random,Deshpande2021tight} showed that at least $\Omega(\log n)$ depth is necessary for anti-concentration, for any architecture with Haar random 2-qubit gates. We also give a simple proof of this fact using the Pauli basis framework in Corollary~\ref{cor:aclowerbound}.
    \item \cite{Dalzell2022random} remarked that, as anti-concentration is proven for two architectures which are two opposite extremes of geometric locality, they conjecture $\Theta(n\log n)$ size (which is $\Theta(\log n)$ depth in our case) to be necessary and sufficient for anti-concentration for any reasonably well-connected architecture.
\end{itemize}
\end{remark}

\begin{remark}\label{remark:randompauli}
The results discussed in Remark~\ref{remark:anticoncentration} only concern Haar random 2-qubit gates. We expect the same results to hold for a fixed 2-qubit gate surrounded by Haar random single qubit gates. It is worth mentioning that while a fixed 2-qubit gate surrounded by random Pauli gates satisfies Definition~\ref{def:architecture}, we do not expect it to satisfy anti-concentration, due to the fact that it does not generate the entire Clifford group when, for example, the 2-qubit gate is a CNOT gate.
\end{remark}

The reason for requiring anti-concentration for our results is because it is closely related to the Fourier weights of random circuits, which is then related to the error of the simulation algorithm.

\begin{definition}[Fourier weight]
The \textbf{Fourier weight} of a random circuit $C$ at degree $k$ is defined as
\begin{equation}
    W_k=2^{2n}\E_C\sum_{s\in\mathsf{P}_n^{d+1}:|s|=k}f(C,s,0^n)^2.
\end{equation}
\end{definition}

Here the $2^{2n}$ factor is a normalization factor that comes from Eq.~\eqref{eq:boundary}. A crucial property for our arguments is that anti-concentration implies that the total Fourier weight is upper bounded by a constant.

\begin{lemma}[Total Fourier weight]\label{lemma:totalfourierweight}
Let $\mc D$ be a distribution over quantum circuits that satisfies anti-concentration and Definition~\ref{def:architecture}. The Fourier weights $\{W_k\}$ satisfy
\begin{enumerate}
    \item $W_0=1$,
    \item $W_k=0$, $\forall 0<k\leq d$,
    \item $\sum_{k\geq d+1}W_k=O(1)$.
\end{enumerate}
\end{lemma}
\begin{proof}
    $W_0=1$ corresponds to the unique all-identity path. Let $s$ be a Pauli path of Hamming weight $k=|s|\in (0,d]$. Then there exists a 2-qubit gate $U$ that contributes a transition amplitude $\lmel{q}{\mc U}{p}$ to $f(C,s,0^n)$, where either $p$ is identity and $q$ is non-identity, or vice versa. In either case we have $\lmel{q}{\mc U}{p}=0$. This implies that $W_k=0$.
    
    To bound the total weight, we start with anti-concentration.
    \begin{equation}
        \begin{aligned}
            O(1)&=\E_{C\sim\mc D}2^n\sum_{x\in\{0,1\}^n}p(C,x)^2\\
            &=\E_{C\sim\mc D}2^n\sum_{x\in\{0,1\}^n}\left(\sum_{s\in\mathsf{P}_n^{d+1}}f(C,s,x)\right)^2\\
            &=\E_{C\sim\mc D}2^n\sum_{x\in\{0,1\}^n}\sum_{s,s'\in\mathsf{P}_n^{d+1}}f(C,s,x)f(C,s',x)\\
            &=\E_{C\sim\mc D}2^n\sum_{x\in\{0,1\}^n}\sum_{s\in\mathsf{P}_n^{d+1}}f(C,s,x)^2\\
            &=2^{2n}\E_{C\sim\mc D}\sum_{s\in\mathsf{P}_n^{d+1}}f(C,s,0^n)^2\\
            &=2^{2n}\E_{C\sim\mc D}\sum_{k\geq 0}\sum_{s\in\mathsf{P}_n^{d+1}:|s|=k}f(C,s,0^n)^2\\
            &=1+\sum_{k\geq d+1} W_k.
        \end{aligned}
    \end{equation}
Here, the first line follows from anti-concentration; the second line follows from the Pauli path integral; the fourth line follows from orthogonality (Lemma~\ref{lemma:orthogonality}); the fifth line follows from Eq.~\eqref{eq:boundarysign}.
\end{proof}

Finally we give a detailed clarification regarding the assumptions we make about the architecture and gate set for our main result.
\begin{remark}\label{remark:assumption}
For our main result Theorem~\ref{thm:mainresult}, we assume Definition~\ref{def:architecture} and anti-concentration as defined in Definition~\ref{def:anticoncentration}.
\begin{itemize}
    \item If the gate set is Haar random 2-qubit gates, no further assumption is needed.
    \item If not, then we further assume that the circuit depth is at least $\Omega(\log n)$. This is because our algorithm requires $\Omega(\log n)$ depth to be efficient, and we cannot rule out the possibility that there is an ensemble of random circuits below log depth that satisfies both Definition~\ref{def:architecture} and \ref{def:anticoncentration}.
\end{itemize}
\end{remark}

\section{Simulating noisy random circuit sampling}
\label{sec:mainresult}
Given a random circuit $C$ and an output $x$, let $p(C,x)=\left|\mel{x}{C}{0^n}\right|^2$ be the ideal output distribution and let $\tilde{p}(C,x)$ be the output distribution of the noisy circuit where $C$ is subject to local depolarizing noise of rate $\gamma$. This section shows the following:

\begin{theorem}[Restatement of Theorem~\ref{thm:mainresult}]
Let $\mc D$ be a distribution over quantum circuits that satisfies anti-concentration and Definition~\ref{def:architecture} (also see Remark~\ref{remark:assumption}). There is a classical algorithm that, on input $C\sim \mc D$, outputs a sample from a distribution that is $\varepsilon$-close to $\tilde{p}(C,x)$ in total variation distance with success probability at least $1-\delta$ over the choice of $C$, in time $\poly(n,1/\varepsilon,1/\delta)$.
\end{theorem}

Our goal is to compute a function $\bar{q}(C,x)$ that achieves small $L_1$ distance
\begin{equation}
    \Delta:=\left\|\tilde{p}-\bar{q}\right\|_1:=\sum_{x\in\{0,1\}^n}\left|\tilde{p}(C,x)-\bar{q}(C,x)\right|
\end{equation}
with high probability. Here $\{\bar{q}(C,x)\}_x$ is not necessarily a distribution, and $\bar{q}(C,x)$ is not necessarily positive (the bar notation indicates that $\bar{q}$ is a quasi-probability distribution). The main result is derived in three steps:
\begin{enumerate}
    \item We use a general sampling-to-computing reduction shown by~\cite{Bremner2017achievingquantum} which says that given the ability to compute $\bar{q}(C,x)$ as well as its marginals, we can sample from a distribution that is $O(\Delta)$-close to $\tilde{p}(C,x)$ with a polynomial overhead. This is discussed in Section~\ref{sec:samplingtocomputing}. It remains to develop an efficient algorithm to compute $\bar{q}(C,x)$ and its marginals.
    \item The algorithm is to approximate $\tilde{p}(C,x)$ by summing its low-degree Fourier coefficients, defined as
    \begin{equation}
        \bar{q}(C,x):=\sum_{s:|s|\leq\ell}\tilde{f}(C,s,x)=\sum_{s:|s|\leq\ell}(1-\gamma)^{|s|}f(C,s,x),
    \end{equation}
    where $\ell$ is to be determined. In Section~\ref{sec:boundtvd} we upper bound the total variation distance $\Delta$ achieved by this approximation. It shows that choosing $\ell=O(\log 1/\varepsilon)$ suffices to achieve $\varepsilon$ total variation distance.
    \item It remains to bound the running time of the algorithm. In Section~\ref{sec:counting} which is the main technical part, we show that each $\bar{q}(C,x)$ can be computed in time $2^{O(\ell)}$. This completes the argument.
\end{enumerate}

\subsection{Bounds for the total variation distance}
\label{sec:boundtvd}
We show that the expected total variation distance square is upper bounded by an exponential decay of the Fourier weights.

\begin{equation}\label{eq:tvdbound}
    \begin{aligned}
        \E_C \left[\Delta^2\right]&\leq 2^n\E_C\sum_{x\in\{0,1\}^n}\left(\tilde p(C,x)-\bar{q}(C,x)\right)^2\\
        &=2^n\E_C\sum_{x\in\{0,1\}^n}\left(\sum_{s:|s|>\ell}(1-\gamma)^{|s|}f(C,s,x)\right)^2\\
        &=2^n\E_C\sum_{x\in\{0,1\}^n}\sum_{s:|s|>\ell}(1-\gamma)^{2|s|}f(C,s,x)^2\\
        &=2^{2n}\E_C\sum_{s:|s|>\ell}(1-\gamma)^{2|s|}f(C,s,0^n)^2\\
        &=\sum_{k>\ell}(1-\gamma)^{2k}W_k.
    \end{aligned}
\end{equation}
Here, the first line follows from Cauchy–Schwarz; the second line is by definition of $\bar{q}$; the third line follows from orthogonality (Lemma~\ref{lemma:orthogonality}); the fourth line follows from Eq.~\eqref{eq:boundarysign}; the fifth line is by definition of Fourier weight.

A simple upper bound can be derived assuming anti-concentration (item 3 from Lemma~\ref{lemma:totalfourierweight}),
\begin{equation}
    \E_C \left[\Delta^2\right]\leq\sum_{k>\ell}(1-\gamma)^{2k}W_k\leq \sum_{k>\ell}(1-\gamma)^{2\ell}W_k\leq O(1)\cdot e^{-2\gamma\ell}.
\end{equation}
By choosing $\ell=O(\log 1/\varepsilon)$ (roughly $\ell\approx\frac{1}{\gamma}\cdot\log 1/\varepsilon$) we can guarantee that $\Delta\leq\varepsilon$ with high probability.

\subsection{Counting and enumerating legal Pauli paths}
\label{sec:counting}
For a given truncation parameter $\ell$, the running time of the algorithm depends on the number of Pauli paths with Hamming weight at most $\ell$, as well as the efficiency for finding and enumerating these paths. A simple argument for bounding the number of paths is as follows. There are $n(d+1)$ locations in the circuit to insert Pauli paths. The total number of ways to insert $\ell$ non-identity Pauli into the Pauli path is at most $\binom{n(d+1)}{\ell}$, and the choice of $X,Y,Z$ for each non-identity gives a $3^\ell$ factor. Therefore the total number of paths with Hamming weight at most $\ell$ is at most
\begin{equation}
    \ell\cdot \binom{n(d+1)}{\ell}\cdot 3^\ell\leq (n d)^{O(\ell)}.
\end{equation}
In this section we show that this bound is a significant overestimate and can be improved to $2^{O(\ell)}$. The key point here is that only the ``legal'' paths matters, and therefore we design an algorithm that only counts and enumerates legal paths.

\begin{definition}[Legal Pauli path]\label{def:legalpath}
For a given circuit architecture, a Pauli path $s=(s_0,s_1,\dots,s_d)$ is legal if the following two conditions are satisfied:
\begin{enumerate}
    \item For all 2-qubit gates in the circuit, its input and output Paulis are either both \texttt{II}, or both not \texttt{II}.
    \item $s_0$ and $s_d$ contains only \texttt{I} and \texttt{Z}.
\end{enumerate}
\end{definition}
The reason for considering legal Pauli paths is that the illegal ones are irrelevant, as they contribute 0 to the Pauli path integral.
\begin{lemma}
    Any illegal Pauli path $s$ gives $f(C,s,x)=0$ for any $C$ and $x$.
\end{lemma}
\begin{proof}
Let $s$ be an illegal Pauli path. Then there are two cases: either the first or the second condition of Definition~\ref{def:legalpath} is violated. If the second condition is violated, then $f(C,s,x)=0$ because the inner product between computational basis states with $s_0$ or $s_d$ equals 0, due to the fact that
\begin{equation}
    \lbraket{x}{s}=\Tr(\ketbra{x}\cdot s)=0,\quad \forall x\in\{0,1\}^n,s\notin\{I/\sqrt{2},Z/\sqrt{2}\}^{\otimes n}.
\end{equation}
If the first condition is violated, then there is a 2-qubit gate $U$ whose input Pauli is \texttt{II} and the output is not \texttt{II}, or vice versa. Then $f(C,s,x)=0$ because the transition amplitude contributed by $U$ equals 0 due to the fact that unitary channel is trace preserving, i.e.
\begin{equation}
\begin{aligned}
    \lmel{p}{\mc U}{q}=\Tr(p U q U^\dag)=0\quad&\text{if }p=I\otimes I/2,q\neq I\otimes I/2,\\
    &\text{or }p\neq I\otimes I/2,q= I\otimes I/2.
\end{aligned}
\end{equation}
\end{proof}

Next we develop arguments to count legal paths. The number of legal Pauli paths up to a given Hamming weight is a combinatorial property that only depends on the circuit architecture, independent of the gate set.

We first give a simple example that counts the number of legal paths with weight $d+1$. Lemma~\ref{lemma:totalfourierweight} says that $d+1$ is the smallest non-zero Hamming weight with legal paths. The result below is interesting by itself, as we will show later that this result gives a simple lower bound on the depth for anti-concentration (Corollary~\ref{cor:aclowerbound}).

\begin{lemma}\label{lemma:smallestweightpaths}
    The number of legal Pauli paths with Hamming weight $d+1$ equals $n\cdot 2^d\cdot 3^{d-1}$. 
\end{lemma}
\begin{proof}
As the Pauli path $s=(s_0,s_1,\dots,s_d)$ has Hamming weight $d+1$, it has to be the case that $|s_i|=1$ for $i=0,\dots,d$. We first choose the location of the non-identity in the first layer $s_0$, which has $n$ choices. Suppose this non-identity Pauli is at the input of some 2-qubit gate $U$. Then the output of $U$ can be either \texttt{IR} or \texttt{RI} (We use \texttt{R} to represent a non-identity), which gives two choices. Repeating this argument for each layer, we know that the number of configurations of locations of non-identities is $n\cdot 2^d$. Finally, the $3^{d-1}$ factor comes from the fact that the non-identity Pauli at the first and last layer has to be $Z$, while each of the other $d-1$ layers has three choices among $X,Y,Z$.
\end{proof}

Next we show that anti-concentration implies the desired $2^{O(\ell)}$ upper bound for the number of legal paths. This bound is clearly tight up to the constant in the exponent, as even the choice of $X,Y,Z$ for a single path of weight $\ell$ gives a $3^\ell$ factor. The problem with the result below is that it does not give an algorithm to find and enumerate the legal paths. This is addressed later.

\begin{lemma}\label{lemma:pathsfromanticoncentration}
    Consider any circuit architecture which satisfies anti-concentration with Haar random 2-qubit gates. For any $\ell\geq d+1$, the total number of legal Pauli paths with Hamming weight at most $\ell$ is upper bounded by $2^{O(\ell)}$.
\end{lemma}
\begin{proof}
We have shown in Lemma~\ref{lemma:totalfourierweight} that anti-concentration implies that $\sum_{k\geq d+1}W_k = O(1)$. Below we give a lower bound on the Fourier weight up to degree $\ell$. Consider any legal Pauli path $s$ with Hamming weight at most $\ell$. We will calculate its contribution to the Fourier weight $2^{2n}\E_C f(C,s,0^n)^2$ as follows.
\begin{equation}\label{eq:pathcontribution}
    \begin{aligned}
        2^{2n}\E_C\left[ f(C,s,0^n)^2\right]&=2^{2n}\E_C\left[\left(\lbraket{x}{s_d}\lmel{s_d}{\mc U_d}{s_{d-1}}\cdots \lmel{s_1}{\mc U_1}{s_0}\lbraket{s_0}{0^n}\right)^2\right]\\
        &=\E_C\left[\lmel{s_d}{\mc U_d}{s_{d-1}}^2\cdots \lmel{s_1}{\mc U_1}{s_0}^2\right]\\
        &=\E_{\mc U_d}\left[\lmel{s_d}{\mc U_d}{s_{d-1}}^2\right]\cdots\E_{\mc U_1} \left[\lmel{s_1}{\mc U_1}{s_0}^2\right]\\
        &=\left(\frac{1}{15}\right)^{G(s)}
    \end{aligned}
\end{equation}
Here the second line follows from the fact that $|\lbraket{x}{s_d}|=|\lbraket{s_0}{0^n}|=\frac{1}{\sqrt{2^n}}$, the third line is due to the independence between different random gates, and the fourth line is due to Lemma~\ref{lemma:haarrandomgate}, where we define
\begin{equation}
    G(s):=\text{the number of 2-qubit gates whose input and output are not }\texttt{II}\text{ in }s.
\end{equation}
The above calculation says that any 2-qubit gate whose input and output are not \texttt{II} contributes a $\frac{1}{15}$ factor to the Fourier weight. A simple bound on $G(s)$ is
\begin{equation}
    \frac{|s|}{4}\leq G(s)\leq |s|,
\end{equation}
where LHS is because each gate corresponds to at most 4 non-identity Paulis, and RHS is because each gate has at least 1 input non-identity Pauli. This implies that
\begin{equation}
    2^{2n}\E_C\left[ f(C,s,0^n)^2\right]\geq \left(\frac{1}{15}\right)^{|s|}.
\end{equation}
Using this we have
\begin{equation}
    \begin{aligned}
        O(1)&=\sum_{k=d+1}^\ell W_k\\
        &=\sum_{k=d+1}^\ell 2^{2n}\E_C\sum_{s\in\mathsf{P}_n^{d+1}:|s|=k}f(C,s,0^n)^2\\
        &\geq\sum_{k=d+1}^\ell \sum_{s\in\mathsf{P}_n^{d+1}:|s|=k}\left(\frac{1}{15}\right)^{|s|}1[s\text{ is legal}]\\
        &\geq \left(\frac{1}{15}\right)^{\ell}\left(\text{Number of legal paths of weight at most }\ell\right),
    \end{aligned}
\end{equation}
which means that the number of legal paths of weight at most $\ell$ is at most $O(1)\cdot 15^\ell.$
\end{proof}

We have remarked earlier that the number of legal paths is a combinatorial property that only depends on the circuit architecture, independent of the gate set. We introduce Haar random 2-qubit gates in Lemma~\ref{lemma:pathsfromanticoncentration} as a proof technique for bounding the Fourier weights. We further show that the above results imply a lower bound on the depth for anti-concentration, which has been shown by~\cite{Dalzell2022random,Deshpande2021tight} using different techniques.

\begin{corollary}\label{cor:aclowerbound}
Consider any circuit architecture which satisfies anti-concentration with Haar random 2-qubit gates, then the circuit depth satisfies $d=\Omega(\log n)$.
\end{corollary}
\begin{proof}
Consider $\ell=d+1$, using Lemma~\ref{lemma:smallestweightpaths} and Lemma~\ref{lemma:pathsfromanticoncentration} we have
\begin{equation}
    n\cdot 2^d\cdot 3^{d-1}\leq O(1)\cdot 15^{d+1},
\end{equation}
which implies that $d=\Omega(\log n)$.
\end{proof}

Next we present the main result of this section, an algorithm for efficiently enumerating low-weight legal Pauli paths.

\begin{lemma}\label{lemma:legalpaths}
    For any $\ell\geq d+1$, the number of legal Pauli paths with Hamming weight at most $\ell$ is at most $n^{\ell/d}\cdot 2^{O(\ell)}$ (the circuit architecture does not need to satisfy anti-concentration). Furthermore there is an efficient algorithm to enumerate the legal paths in time $n^{\ell/d}\cdot 2^{O(\ell)}$ and memory $\tilde{O}(n d)$.
\end{lemma}
The proof of Lemma~\ref{lemma:legalpaths} is deferred to the end of this section. Next we discuss its relationship with the above results.

First, it appears that Lemma~\ref{lemma:legalpaths} is not tight as it has an additional $n^{\ell/d}$ factor compared with Lemma~\ref{lemma:pathsfromanticoncentration}. In fact this is not the case, due to the fact that Lemma~\ref{lemma:pathsfromanticoncentration} assumes anti-concentration, which by Corollary~\ref{cor:aclowerbound} means that Lemma~\ref{lemma:pathsfromanticoncentration} only holds when $d=\Omega(\log n)$. Note that in this case
\begin{equation}
    n^{\ell/d}=e^{\frac{\ell}{d}\cdot \log n}=2^{O(\ell)},
\end{equation}
so in the anti-concentration regime Lemma~\ref{lemma:legalpaths} gives the same asymptotic result as Lemma~\ref{lemma:pathsfromanticoncentration}, which is tight up to the constant in the exponent. 

Second, when $\ell=O(d)$, Lemma~\ref{lemma:legalpaths} gives $\poly(n)\cdot 2^{O(d)}$. 
Therefore compared with Lemma~\ref{lemma:smallestweightpaths} we conclude that Lemma~\ref{lemma:legalpaths} with $\ell=O(d)$ is tight up to the constant in the exponent, regardless of whether anti-concentration holds.\\

\noindent\textbf{Proof of Lemma~\ref{lemma:legalpaths}.} We prove Lemma~\ref{lemma:legalpaths} in the rest of this section. We will enumerate legal Pauli paths $s=(s_0,s_1,\dots,s_d)$ using the following method.
\begin{enumerate}
    \item For each $d+1\leq k\leq\ell$, choose the Hamming weight $w_0,\dots,w_d$ for each layer, such that $w_0+\cdots +w_d=k$.
    \item Choose the configuration (positions of identities and non-identities) for each layer.
    \item Choose $X/Y/Z$ for each non-identity.
\end{enumerate}

The following is a detailed counting argument and enumeration method for the legal Pauli paths. Consider a fixed total Hamming weight $d+1\leq k\leq\ell$.
\begin{enumerate}
    \item Choose the Hamming weight $w_0,\dots,w_{d}$ for each layer, such that the total weight is $k$. The number of choices equals the number of solutions to the equation $w_0+w_2+\cdots +w_d=k$ ($w_i\geq 1$), which equals to $\binom{k-1}{d}\leq 2^{k-1}$. The enumeration of such solutions can be achieved using a combinations enumerator which efficiently enumerates all combinations of choosing $d$ objects from $k-1$ objects, with memory cost $\tilde{O}(d)$. Note that not all solutions correspond to legal Pauli paths; the illegal ones will be rejected later.
    
    \item For each Hamming weight configuration $w_0,\dots,w_d$, let $t$ be the index of the layer with smallest Hamming weight (if there are tiebreaks, choose the smallest $t$). As the total weight is $k$, we know that $w_t\leq k/d$. Next we enumerate the configuration (locations of non-identities) of this layer. The number of choices is $\binom{n}{w_t}\leq n^{k/d}$ and can be enumerated using a combinations enumerator. We can store a configuration of a layer using $n$ bits.
    
    \item We choose the configurations for the other layers in a way that evolves the $t$-th layer both forwards and backwards. For example, consider choosing the configuration for the $t+1$-th layer, conditioned on a given configuration for the $t$-th layer. Consider the layer of 2-qubit gates that connects the $t$-th layer of the Pauli path with the $t+1$-th layer of the Pauli path. Those 2-qubit gates that have input \texttt{II} have to have output \texttt{II}. The number of 2-qubit gates whose input is not \texttt{II} is at most $w_t$. For each of these gates, its output can be \texttt{IR}, \texttt{RI}, or \texttt{RR} (We use \texttt{R} to represent a non-identity). So there are at most $3^{w_t}$ configurations for the $t+1$-th layer. Not all of these configurations satisfy the constraint that the $t+1$-th layer has Hamming weight $w_{t+1}$. So within these (at most) $3^{w_t}$ configurations, we reject those that do not have weight $w_{t+1}$. Repeating this procedure for the next layer, we have that the number of configurations for the $t+2$-th layer is at most $3^{w_{t+1}}$, conditioned on a given configuration for the $t+1$-th layer. Using the same argument but evolve backward from the $t$-th layer, the number of configurations for the $t-1$-th layer is at most $3^{w_{t}}$, and the number of configurations for the $t-2$-th layer is at most $3^{w_{t-1}}$ and so on.
    \item Repeat the above argument for $t+1,t+2,\dots,d$ as well as $t-1,t-2,\dots,0$. The total number of configurations for the entire Pauli path (conditioned on a given partition $w_0,\dots,w_d$ and a given configuration for the $t$-th layer) is at most $3^{\sum_i w_i}=3^{k}$. The memory cost for enumerating a configuration for the entire circuit is at most $\tilde{O}(n d)$.
    \item Replace each \texttt{R} with $X,Y,Z$ (except for the first and last layer, where \texttt{R} is only replaced with $Z$), giving another $3^{k}$ factor.
\end{enumerate}
Taking into account all factors in the above steps, the total number of legal paths of Hamming weight at most $\ell$ (and the total running time of the enumeration algorithm) is at most
    \begin{equation}
        \sum_{k=d+1}^\ell 2^{k-1}\cdot n^{k/d}\cdot 3^k\cdot 3^k\leq \ell\cdot n^{\ell/d}\cdot 18^\ell=n^{\ell/d}\cdot 2^{O(\ell)}.
    \end{equation}

\subsection{Putting everything together}
\label{sec:samplingtocomputing}
Summarizing the main results of the previous section, we have the following.
\begin{lemma}\label{lemma:marginal}
    Consider the same assumptions as our main result (Remark~\ref{remark:assumption}) and fix a truncation parameter $\ell$. There is an algorithm that computes the function $\bar{q}(C,x)=\sum_{s:|s|\leq\ell}(1-\gamma)^{|s|}f(C,s,x)$ and its marginals in time $n d\cdot 2^{O(\ell)}$. Here by marginal we mean $\sum_{i\in T}\sum_{x_i\in\{0,1\}}\bar{q}(C,x_1,\dots,x_n)$ for any $T\subseteq [n]$.
\end{lemma}
\begin{proof}
As circuit depth $d=\Omega(\log n)$, Lemma~\ref{lemma:legalpaths} says that for any $x\in\{0,1\}^n$, $\bar{q}(C,x)$ can be computed in time $n d\cdot 2^{O(\ell)}$ using the enumeration algorithm, as there are $2^{O(\ell)}$ paths and each path takes $O(n d)$ time to compute. To compute a certain marginal $\sum_{i\in T}\sum_{x_i\in\{0,1\}}\bar{q}(C,x_1,\dots,x_n)$, note that we cannot straightforwardly compute each $\bar{q}(C,x_1,\dots,x_n)$ and sum them up because it has an additional factor $2^{|T|}$. However, the marginal can be easily computed by exchanging the summation order,
\begin{equation}
\begin{aligned}
    \sum_{i\in T}\sum_{x_i\in\{0,1\}}\bar{q}(C,x_1,\dots,x_n)&=\sum_{i\in T}\sum_{x_i\in\{0,1\}}\sum_{s:|s|\leq\ell}(1-\gamma)^{|s|}f(C,s,x_1,\dots,x_n)\\
    &=\sum_{s:|s|\leq\ell}(1-\gamma)^{|s|}\left(\sum_{i\in T}\sum_{x_i\in\{0,1\}}f(C,s,x_1,\dots,x_n)\right).
\end{aligned}
\end{equation}
The statement follows from the fact that the summation in the bracket can be computed in time $O(n d)$. This is because
\begin{equation}
\begin{aligned}
    \sum_{i\in T}\sum_{x_i\in\{0,1\}}f(C,s,x_1,\dots,x_n)&=\sum_{i\in T}\sum_{x_i\in\{0,1\}}\lbraket{x}{s_d}\lmel{s_d}{\mc U_d}{s_{d-1}}\cdots \lmel{s_1}{\mc U_1}{s_0}\lbraket{s_0}{0^n}\\
    &=\lbraket{x'}{s_d}\lmel{s_d}{\mc U_d}{s_{d-1}}\cdots \lmel{s_1}{\mc U_1}{s_0}\lbraket{s_0}{0^n},
\end{aligned}
\end{equation}
where
\begin{equation}
    \lbraket{x'}{s_d}=\Tr(s_d\cdot \bigotimes_{j\notin T}\ketbra{x_j}\bigotimes_{i\in T}I_i).
\end{equation}
\end{proof}

Lemma~\ref{lemma:marginal} allows us to use the standard reduction of sampling from a probability distribution via computing its marginals. An issue here is that $\bar{q}(C,x)$ is not necessarily a distribution; it is only guaranteed to be close to $\tilde{p}(C,x)$ in $L_1$ norm. We use the following result of \cite{Bremner2017achievingquantum} which allows us to sample from a distribution that is close to $\tilde{p}(C,x)$.

\begin{lemma}[Lemma 10 in \cite{Bremner2017achievingquantum}]
\label{lemma:samplingtocomputing}
    Let $p$ be a probability distribution on $\{0,1\}^n$. Assume there is an oracle that computes a function $\bar{q}:\{0,1\}^n\to\mathbb{R}$ as well as its marginals, such that $\left\|p-\bar{q}\right\|_1\leq \delta$. Then there is an algorithm that samples from a probability distribution $q$ using $O(n)$ calls to the oracle, such that $\left\|p-q\right\|_1\leq 4\delta/(1-\delta)$.
\end{lemma}

\noindent\textbf{Proof of Main result.} In Section~\ref{sec:boundtvd} we have shown that $\E_C\left[\Delta^2\right]\leq O(1)\cdot e^{-2\gamma\ell}$. By Markov's inequality,
\begin{equation}
    \Pr[\Delta\geq \frac{1}{\sqrt{\delta}}\sqrt{\E\left[ \Delta^2\right]}]=\Pr[\Delta^2\geq \frac{1}{\delta}\E \left[\Delta^2\right]]\leq \delta.
\end{equation}
Therefore, with probability at least $1-\delta$ over random circuit $C$, we have
\begin{equation}\label{eq:tvdboundwhp}
    \Delta\leq \frac{1}{\sqrt{\delta}}\sqrt{\E\left[ \Delta^2\right]}\leq \frac{O(1)}{\sqrt{\delta}}e^{-\gamma\ell}.
\end{equation}

Using Lemma~\ref{lemma:marginal} and Lemma~\ref{lemma:samplingtocomputing}, for those circuits that satisfy Eq.~\eqref{eq:tvdboundwhp} we can sample from a probability distribution that is $O(1)\cdot \Delta$-close to $\tilde{p}(C,x)$ in total variation distance. Let $\varepsilon$ be the desired total variation distance, then \begin{equation}
    \frac{O(1)}{\sqrt{\delta}}e^{-\gamma\ell}\leq\varepsilon\text{ is satisfied when }\ell\geq \frac{1}{\gamma}\log\frac{O(1)}{\varepsilon\cdot \sqrt{\delta}}.
\end{equation}
Obtaining one sample requires $O(n)$ calls to the algorithm in Lemma~\ref{lemma:marginal}. Assuming circuit depth is $d\leq\poly(n)$, the total running time for obtaining one sample is $n\cdot n d\cdot 2^{O(\ell)}=\poly(n)\cdot \left(O(1)/(\varepsilon\cdot\sqrt{\delta})\right)^{O(1/\gamma)}=\poly(n,1/\varepsilon,1/\delta)$.

\subsection{Statistical indistinguishability}
\label{sec:statisticalindistinguishability}

Next we show that our main result implies statistical indistinguishability. We first recall the basic notions and then give a proof of Corollary~\ref{cor:indistinguishability}.

Given two known probability distributions $p,q$ over the same finite alphabet ($\{0,1\}^n$ in our case), and given $M$ samples from either $p$ or $q$, we would like to tell which is the case with high success probability. That is, two known distributions $p$ and $q$ are statistically distinguishable if there is an algorithm $\mc A$ (with unbounded running time) that, on input $x_1,\dots,x_M\sim \mc D$,
\begin{itemize}
    \item if $\mc D=p$, $\mc A$ returns ``$\mc D=p$'' with probability at least $\frac{2}{3}$;
    \item if $\mc D=q$, $\mc A$ returns ``$\mc D=q$'' with probability at least $\frac{2}{3}$.
\end{itemize}
Two known distributions $p$ and $q$ are statistically indistinguishable with $M$ samples if there is no algorithm $\mc A$ that satisfies the above condition. We use the following well-known fact that closeness in total variation distance implies statistical indistinguishability.

\begin{lemma}
    Two known distributions $p$ and $q$ are statistically indistinguishable with $M$ samples if 
    \begin{equation}
        \frac{1}{2}\left\|p-q\right\|_1<\frac{1}{3M}.
    \end{equation}
\end{lemma}

In the context of random circuit sampling, statistical distinguishability is similarly defined with an additional averaging over the random circuit.

\begin{definition}[Statistical distinguishability]
\label{def:distinguishability}
For a random circuit $C$, let $\tilde{p}(C,x)$ be the noisy RCS output distribution and let $q(C,x)$ be a classical mock-up distribution (the output distribution of a classical simulation algorithm). $\tilde{p}(C,x)$ is statistically distinguishable from $q(C,x)$ with $M$ samples if there is an algorithm $\mc A$ with input $C$ as well as $x_1,\dots,x_M\in\{0,1\}^n$ and output one of $\{\text{noisy RCS},\text{mock-up}\}$ (with unbounded running time) such that
\begin{itemize}
    \item $\E_C \Pr_{x_1,\dots,x_M\sim \tilde{p}(C)}[\mc A(C,x_1,\dots,x_M)=\text{noisy RCS}]\geq\frac{2}{3}$,
    \item $\E_C \Pr_{x_1,\dots,x_M\sim q(C)}[\mc A(C,x_1,\dots,x_M)=\text{noisy RCS}]\leq\frac{1}{3}$.
\end{itemize}
\end{definition}

\noindent\textbf{Proof of Corollary~\ref{cor:indistinguishability}.} In order to prove statistical indistinguishability it suffices to show that
\begin{equation}
    \E_C \left[\left\|\tilde{p}(C)^{\otimes M},q(C)^{\otimes M}\right\|_1\right] < \frac{1}{3}.
\end{equation}
Our main result says that $\left\|\tilde{p}(C)-q(C)\right\|_1\leq\varepsilon$ with probability at least $1-\delta$ over $C$. Call those $C$ that satisfy $\left\|\tilde{p}(C)-q(C)\right\|_1\leq\varepsilon$ good, and the rest bad. We have
\begin{equation}
\begin{aligned}
    \E_C \left[\left\|\tilde{p}(C)^{\otimes M},q(C)^{\otimes M}\right\|_1\right]&\leq \E_C \left[\left\|\tilde{p}(C)^{\otimes M},q(C)^{\otimes M}\right\|_1|C\text{ is good}\right]+\Pr[C\text{ is bad}]\\
    &\leq \E_C \left[M\cdot \left\|\tilde{p}(C),q(C)\right\|_1|C\text{ is good}\right]+\delta\\
    &\leq M\varepsilon + \delta,
\end{aligned}
\end{equation}
where the first line follows from the law of total expectation and the second line follows from subadditivity of total variation distance with respect to tensor product. Therefore, statistical indistinguishability is guaranteed by choosing $\varepsilon=0.01/M$ and $\delta=0.01$, which gives running time $\poly(n,M)$ in our algorithm.

\section{Generalizing to an approximation of Google and USTC's gate sets}
\label{sec:gateset}
In this section we discuss the role of gate sets in our main result. Assuming anti-concentration holds and at least $\Omega(\log n)$ depth, then in fact the only place in the proof of our main result where the gate set is relevant is in the third line of Eq.~\eqref{eq:tvdbound}. It uses a property of the Pauli paths called orthogonality (Lemma~\ref{lemma:orthogonality}), which follows from a property of the gate set which we call gate-set orthogonality (Lemma~\ref{lemma:gatesetorthogonality}). Gate-set orthogonality says that in the Pauli basis, if we consider averaging over two copies of a random gate in the gate set, then it effectively forces the input Pauli to be identical across the two copies. Lemma~\ref{lemma:gatesetorthogonality} shows that this holds as long as the gate set is closed under random Pauli.

However, in Google and USTC's experiments~\cite{Arute2019,Wu2021Strong,ZHU2022Quantum} this condition is violated. They considered random circuits with fixed 2-qubit gates and random single-qubit gates, where the 2-qubit gates are called $\fsim$ and are roughly parameterized as follows,
\begin{equation}
    \fsim(\omega_1,\omega_2,\omega_3)=\begin{bmatrix}
    1 & 0 & 0 & 0\\
    0 & 0 & e^{-i\omega_1} & 0\\
    0 & e^{-i\omega_2} & 0 & 0\\
    0 & 0 & 0 & e^{-i\omega_3}
    \end{bmatrix}.
\end{equation}
These angles are site-dependent and are determined by benchmarking experiments. The single-qubit gates are chosen randomly\footnote{Google's single qubit gates $V$ are not independent across each layer; neighboring layers does not repeat. This is still covered by Lemma~\ref{lemma:googlegateset} as it holds even for any fixed $V\in\{\sqrt{X},\sqrt{Y},\sqrt{W}\}$.} from $\{\sqrt{X},\sqrt{Y},\sqrt{W}\}$, where $W=(X+Y)/\sqrt{2}$. 

Here we consider a related gate set shown in LHS of Fig.~\ref{fig:fsim} where the main difference is that we insert random $Z$ rotations. The $\fsim$ gates have a special property that allows us to borrow randomness from $R_Z(\theta_3)$, $R_Z(\theta_4)$ and create additional random gates as $R_Z(\theta_5)$, $R_Z(\theta_6)$, leading to the equivalent gate set in RHS of Fig.~\ref{fig:fsim}. This is because of the following commutation property. By definition, we can check that for any angles $\theta_1,\theta_2,\omega=(\omega_1,\omega_2,\omega_3)$,

\begin{equation}
    R_Z(\theta_1)\otimes R_Z(\theta_2) \cdot\fsim(\omega) = \fsim(\omega)\cdot R_Z(\theta_2)\otimes R_Z(\theta_1).
\end{equation}

\begin{figure}[t]
    \centering
    \includegraphics[width=\textwidth]{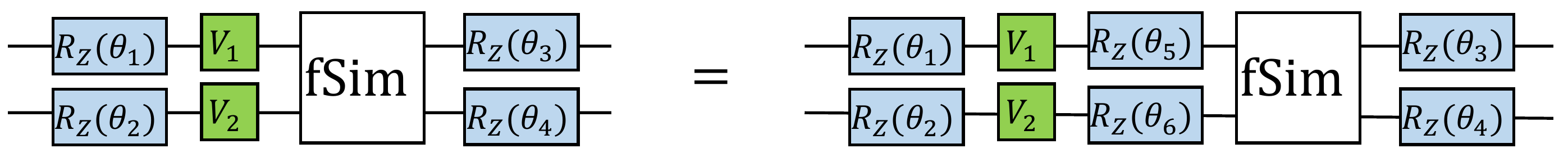}
    \caption{A gate set related to Google and USTC's experiments, for which our main result holds. LHS: the gate set consists of a fixed $\fsim$ gate surrounded by random gates from $\{\sqrt{X},\sqrt{Y},\sqrt{W}\}$ as well as random $Z$ rotations. RHS: this is equivalent to LHS due to a special property of the $\fsim$ gates.}
    \label{fig:fsim}
\end{figure}

Therefore we can consider the effective single qubit gate set $R_Z(\theta_1) V R_Z(\theta_2)$, $V\in\{\sqrt{X},\sqrt{Y},\sqrt{W}\}$. By direct calculation, we can verify that this single-qubit gate set is invariant under random Pauli and thus satisfies gate-set orthogonality.

\begin{lemma}\label{lemma:googlegateset}
    Let $\mc D$ be a distribution over single-qubit unitary defined as $R_Z(\theta_1)V R_Z(\theta_2)$ where $\theta_1,\theta_2\sim[-\pi,\pi]$ and $V\sim\{\sqrt{X},\sqrt{Y},\sqrt{W}\}$. Then for any $P,Q\in\{I,X,Y,Z\}$ such that $P\neq Q$, we have
    \begin{equation}
        \E_{U\sim \mc D}\left[U P U^\dag\otimes U Q U^\dag\right] = 0.
    \end{equation}
\end{lemma}

This implies the orthogonality condition in Lemma~\ref{lemma:orthogonality} which implies that our main result holds. An interesting open question is whether orthogonality is necessary for our main result, and whether our main result holds for the exact gate sets used in Google and USTC's experiments.

\section*{Acknowledgements}
We thank Scott Aaronson, Sergio Boixo, Adam Bouland, and Bill Fefferman for helpful discussions and feedback on the manuscript. We would also like to thank the Simons Institute for the Theory of Computing and the NSF Workshop on Quantum Advantage and Next Steps, where part of this work was done. D.A.~is supported by ISF grant number 0399494-1721/17, by Simons grant number 385590, and by Quantum ISF grant number 2137/19. X.G.~is supported by the Postdoctoral Fellowship in Quantum Science of the MPHQ, the Templeton Religion Trust Grant No. TRT 0159, and by the Army Research Office under Grant No. W911NF1910302 and MURI Grant No. W911NF2010082. Z.L., Y.L. and U.V. are supported by Vannevar Bush faculty fellowship N00014-17-1-3025, MURI Grant FA9550-18-1-0161, DOE
NQISRC Quantum Systems Accelerator grant FP00010905, and NSF QLCI Grant No.~2016245. Y.L. is also supported by NSF award DMR-1747426.

\printbibliography

\appendix
\section{Refuting XQUATH for sublinear depth random circuits}\label{sec:xquath}
Here we give a formal refutation of the XQUATH conjecture of \cite{Aaronson2020on} using the Pauli basis framework (a similar argument was first sketched in \cite{Gao2021Limitations}; here we give a more formal treatment). In this section we only consider Haar random 2-qubit gates.

The XQUATH conjecture is about the hardness of estimating the output probability $p(C,0^n)$ of an ideal random circuit $C$. It says that no efficient classical algorithm can achieve a slightly better variance compared with the trivial algorithm of outputting $\frac{1}{2^n}$.

\begin{conjecture}[XQUATH~\cite{Aaronson2020on}]
Let $\mc D$ be a distribution over quantum circuits. There is no polynomial-time classical algorithm that takes as input a quantum circuit $C\sim\mc D$ and produces a number $q(C,0^n)$ such that
\begin{equation}
    \xq:=2^{2n}\left(\E_{C\sim\mc D}\left[\left(p(C,0^n)-\frac{1}{2^n}\right)^2\right]-\E_{C\sim\mc D}\left[\left(p(C,0^n)-q(C,0^n)\right)^2\right]\right)=\Omega\left(\frac{1}{2^n}\right).
\end{equation}
\end{conjecture}

The intuition behind this conjecture is the Feynman path integral in the computational basis (Eq.~\eqref{eq:feynmanpathcomputational}). There are $2^{O(n d)}$ Feynman paths that are uniform, meaning that each path gives the same contribution on average. Therefore, a polynomial time classical algorithm cannot obtain a good estimate by calculating polynomial paths. In contrast, the Pauli path integral is highly non-uniform and the contribution of a path decays exponentially with the Hamming weight (Eq.~\eqref{eq:pathcontribution}).

Remarkably, we show that using the Pauli path integral, a single path suffices to refute this conjecture below linear depth.

\begin{theorem}\label{thm:xquath}
    Let $\mc D$ be a distribution over quantum circuits with Haar random 2-qubit gates. Then there exists an algorithm that, on input $C$, outputs a number $q(C,0^n)$ in time $O(nd)$ that achieves
    \begin{equation}
        \xq=\left(\frac{1}{15}\right)^d.
    \end{equation}
    Therefore XQUATH is false for random circuits with depth $d=o(n)$.
\end{theorem}
\begin{proof}
    On input $C=U_d\cdots U_1$, define the algorithm as computing
    \begin{equation}
        q(C,0^n):=\frac{1}{2^n}+\lbraket{0^n}{s^*}\lmel{s^*}{\mc U_d}{s^*}\cdots \lmel{s^*}{\mc U_1}{s^*}\lbraket{s^*}{0^n}=\frac{1}{2^n}+f(C,\vec{s^*},0^n)
    \end{equation}
    where $s^*=\frac{1}{\sqrt{2^n}}Z_1\otimes I^{\otimes n-1}$ ($Z_1$ acts on the first qubit). This takes time $O(nd)$. Then
    \begin{equation}\label{eq:XQUATH}
        \begin{aligned}
        \xq&=2^{2n}\E_{C\sim\mc D}\left(\frac{1}{2^{2n}}-\frac{2}{2^n}p(C,0^n)-q(C,0^n)^2+2p(C,0^n)q(C,0^n)\right)\\
        &=2^{2n}\E_{C\sim\mc D}\left(-\frac{1}{2^{2n}}-q(C,0^n)^2+2p(C,0^n)q(C,0^n)\right)\\
        &=2^{2n}\E_{C\sim\mc D}\left(-\frac{2}{2^{2n}}-f(C,\vec{s^*},0^n)^2+2p(C,0^n)q(C,0^n)\right)\\
        &=2^{2n}\E_{C\sim\mc D}\left(-f(C,\vec{s^*},0^n)^2+2p(C,0^n)f(C,\vec{s^*},0^n)\right)\\
        &=2^{2n}\E_{C\sim\mc D}\left(-f(C,\vec{s^*},0^n)^2+2f(C,\vec{s^*},0^n)^2\right)\\
        &=2^{2n}\E_{C\sim\mc D}f(C,\vec{s^*},0^n)^2\\
        &=\left(\frac{1}{15}\right)^d.
        \end{aligned}
    \end{equation}
Here, the first line is by calculation; the second line is because $\E_{C\sim \mc D}[p(C,0^n)]=\frac{1}{2^n}$; the third line is because $\E_{C\sim \mc D}[f(C,\vec{s^*},0^n)]=0$ which follows from orthogonality with the all-identity path; the fourth line is again because $\E_{C\sim \mc D}[p(C,0^n)]=\frac{1}{2^n}$; the fifth line follows from orthogonality (Lemma~\ref{lemma:orthogonality}); the final step follows from the fact that there are exactly $d$ 2-qubit gates that has a non-identity at the input and output, and each gate contributes a $\frac{1}{15}$ factor due to Lemma~\ref{lemma:haarrandomgate}.
\end{proof}

As XQUATH is closely related to the XEB test, next we give a similar result by applying the above algorithm to XEB. Note that in actual experiments the XEB should be viewed as a statistical test, and our main result already implies that no such tests can distinguish between noisy RCS and the efficient classical algorithm in our main result. Thus the discussions below are for demonstration purposes, and for simplicity we only consider the expected value of XEB, and show that one Pauli path already suffices to achieve $2^{-O(d)}$ XEB.

For a random circuit $C$, let $p(C,x)$ be the output distribution of $C$, and let $q(C,x)$ be the output distribution of the noisy implementation of $C$ or a simulation algorithm. The expected value of the linear cross entropy is defined as
\begin{equation}
    \xeb := 2^n\E_C\sum_{x\in\{0,1\}^n}p(C,x)q(C,x) - 1.
\end{equation}

First we consider noisy random circuits with the same model as in our main result. A useful property is that the XEB of noisy random circuits can be viewed as the Fourier weight polynomial.
\begin{equation}
    \begin{aligned}
        \xeb &= 2^n\E_C\sum_{x\in\{0,1\}^n}p(C,x)q(C,x) - 1\\
        &=2^n\E_C\sum_{x\in\{0,1\}^n}\sum_{s}(1-\gamma)^{|s|}f(C,s,x)^2 - 1\\
        &=2^{2n}\E_C\sum_{s}(1-\gamma)^{|s|}f(C,s,0^n)^2 - 1\\
        &=\sum_{k>0}(1-\gamma)^{k} W_k.
    \end{aligned}
\end{equation}
Here, the second line follows from the Pauli path integral and orthogonality (Lemma~\ref{lemma:orthogonality}); the third line is by Eq.~\eqref{eq:boundarysign}; the fourth line is by definition of Fourier weight and the fact that $W_0=1$.

\begin{theorem}
    Assuming anti-concentration, the linear cross entropy of a noisy random circuit with $\gamma$ depolarizing noise satisfies
    \begin{equation}
        (1-\gamma)^{d+1}\cdot n\cdot 2^d\cdot 3^{d-1}\cdot \left(\frac{1}{15}\right)^d \leq\xeb\leq O(1)\cdot e^{-\gamma d}.
    \end{equation}
    Note that anti-concentration is only used for the upper bound; the lower bound does not require anti-concentration.
\end{theorem}
\begin{proof}
    For the upper bound we use the upper bound on total Fourier weight (Lemma~\ref{lemma:totalfourierweight}),
    \begin{equation}
        \xeb=\sum_{k\geq d+1}(1-\gamma)^{k} W_k\leq O(1)\cdot (1-\gamma)^{d+1}\leq O(1)\cdot e^{-\gamma d}.
    \end{equation}
    The lower bound follows from $\xeb\geq (1-\gamma)^{d+1}W_{d+1}$. We have $W_{d+1}=n\cdot 2^d\cdot 3^{d-1}\cdot \left(\frac{1}{15}\right)^d$ because we have shown in Lemma~\ref{lemma:smallestweightpaths} that the number of legal Pauli paths at degree $d+1$ equals $n\cdot 2^d\cdot 3^{d-1}$; each of them contributes $\left(\frac{1}{15}\right)^d$ to the Fourier weight.
\end{proof}

Next, consider a classical algorithm which samples from the same distribution as in the proof of Theorem~\ref{thm:xquath}, which is the following distribution
\begin{equation}
    q(C,x)=\frac{1}{2^n}+\lbraket{x}{s^*}\lmel{s^*}{\mc U_d}{s^*}\cdots \lmel{s^*}{\mc U_1}{s^*}\lbraket{s^*}{0^n}=\frac{1}{2^n}+f(C,\vec{s^*},x)
\end{equation}
where $s^*=\frac{1}{\sqrt{2^n}}Z_1\otimes I^{\otimes n-1}$ ($Z_1$ acts on the first qubit). $\{q(C,x)\}_{x\in\{0,1\}^n}$ is a probability distribution because $\left|f(C,\vec{s^*},x)\right|\leq\frac{1}{2^n}$ and $\sum_{x\in\{0,1\}^n}f(C,\vec{s^*},x)=0$.

The algorithm only samples the first qubit non-trivially, and uniformly on all other qubits.
\begin{theorem}
    There exists an efficient classical algorithm that, given a random circuit, outputs a sample in time $O(n d)$ that achieves
    \begin{equation}
        \xeb=\left(\frac{1}{15}\right)^d.
    \end{equation}
\end{theorem}
\begin{proof}
\begin{equation}
    \begin{aligned}
        \xeb&= 2^n\E_C\sum_{x\in\{0,1\}^n}p(C,x)f(C,s^*,x)\\
        &= 2^n\E_C\sum_{x\in\{0,1\}^n}\sum_{s=(s_0,\dots,s_d)\in\mathsf{P}_n}f(C,s,x)f(C,s^*,x)\\
        &=2^n\E_C\sum_{x\in\{0,1\}^n}f(C,s^*,x)^2\\
        &=2^{2n}\E_Cf(C,s^*,0^n)^2=\left(\frac{1}{15}\right)^d.
    \end{aligned}
\end{equation}
Here, the first line is by definition of $q(C,x)$; the second line is by Pauli path integral; the third line follows from orthogonality (Lemma~\ref{lemma:orthogonality}); the fourth line is by Eq.~\eqref{eq:boundarysign}.
\end{proof}

\section{Improved bounds on the convergence of noisy random circuits to the uniform distribution}
\label{sec:convergetouniform}

In this section we develop improved bounds on the total variation distance between the output distribution of noisy random circuits and the uniform distribution, focusing on Haar random 2-qubit gates and depolarizing noise. The result can be directly applied to other depolarizing-like noise models as in~\cite{Gao2018efficient,Deshpande2021tight}.

Let $U$ denote the uniform distribution. In \cite{Deshpande2021tight}, the authors prove that the average total variation distance is lower bounded as
\begin{equation}
    \E_C\left[\frac{1}{2}\left\|\tilde{p}(C)-U\right\|_1\right]\geq\frac{(1-\gamma)^{2d}}{4\cdot 30^d}.
\end{equation}
Here we improve this bound by applying the Fourier weight and Pauli path integral technique. Note that in order to match with the notation in \cite{Deshpande2021tight}, we remove the layer of noise channels applied before the first layer of gates in Fig.~\ref{fig:rcs} (b).

\begin{theorem}
Let $\mc D$ be a distribution over quantum circuits on any parallel circuit architecture with Haar random 2-qubit gates. Let $\tilde{p}(C)$ be the output distribution of $C$ subject to depolarizing noise with error rate $\gamma$ on each qubit after each layer of gates, then we have 
\begin{equation}
\E_{C\sim\mc D}\left[\frac{1}{2}\left\|\tilde{p}(C)-U\right\|_1\right]\geq\frac{1}{12}\cdot (1-\gamma)^{2d}\cdot\left(\frac{2}{5}\right)^d.
\end{equation}
\end{theorem}

\begin{proof}
Let $\delta:=\frac{1}{2}\left\|\tilde{p}(C)-U\right\|_1$ be the total variation distance, and let
\begin{equation}
    \tilde{p}_0:=\sum_{y\in\{0,1\}^{n-1}}\tilde{p}(C,0y)
\end{equation}
be the marginal output probability of the first qubit being 0. As also noted by \cite{Deshpande2021tight}, the total variation distance can be lower bounded as $\delta\geq\left|\tilde{p}_0-\frac{1}{2}\right|\geq\left(\tilde{p}_0-\frac{1}{2}\right)^2$.

Following the proof of Lemma~\ref{lemma:marginal}, the marginal output probability can be written as the Pauli path integral
\begin{equation}
    \tilde{p}_0=\sum_{s=(s_0,\dots,s_d)\in\mathsf{P}_n^{d+1}}g(C,s)
\end{equation}
where
\begin{equation}
\begin{aligned}
    g(C,s)&:=\lbraket{0I^{\otimes n-1}}{s_d}\lmel{s_d}{\mc E^{\otimes n}\mc U_d}{s_{d-1}}\cdots \lmel{s_1}{\mc E^{\otimes n}\mc U_1}{s_0}\lbraket{s_0}{0^n}\\
    &=(1-\gamma)^{|s_d|+\cdots +|s_1|}\lbraket{0I^{\otimes n-1}}{s_d}\lmel{s_d}{\mc U_d}{s_{d-1}}\cdots \lmel{s_1}{\mc U_1}{s_0}\lbraket{s_0}{0^n}.
\end{aligned}
\end{equation}
The trivial all-identity path contributes $\frac{1}{2}$ to $\tilde{p}_0$. Therefore,
\begin{equation}\label{eq:tvdpaths}
    \begin{aligned}
    \E_{C\sim\mc D}\left[\frac{1}{2}\left\|\tilde{p}(C)-U\right\|_1\right]&\geq\E_{C\sim\mc D}\left(\tilde{p}_0-\frac{1}{2}\right)^2\\
    &=\E_{C\sim\mc D}\left(\sum_{s:|s|>0}g(C,s)\right)^2\\
    &=\E_{C\sim\mc D}\sum_{s:|s|>0}g(C,s)^2,
    \end{aligned}
\end{equation}
where the third line is by orthogonality (Lemma~\ref{lemma:orthogonality}). To lower bound the above sum, we consider all Pauli paths of weight $d+1$. Any such path $s$ that gives a non-zero contribution will have Fourier weight
\begin{equation}
    \E_{C\sim\mc D}g(C,s)^2=(1-\gamma)^{2d} \cdot 2^{n-2}\cdot \left(\frac{1}{15}\right)^d\cdot \frac{1}{2^n}=\frac{(1-\gamma)^{2d}}{4\cdot 15^d}.
\end{equation}
It remains to count the number of such paths. We have shown in Lemma~\ref{lemma:smallestweightpaths} that the number of legal Pauli paths with weight $d+1$ equals $n\cdot 2^d\cdot 3^{d-1}$. However, here we lose a factor of $n$ because the final layer has to be $s_d=\frac{1}{\sqrt{2^n}}Z_1\otimes I^{\otimes n-1}$ ($Z_1$ acts on the first qubit) and therefore the path has a fixed ending point, giving $2^d\cdot 3^{d-1}$ paths in total. This gives
\begin{equation}
    \begin{aligned}
    \E_{C\sim\mc D}\left[\frac{1}{2}\left\|\tilde{p}(C)-U\right\|_1\right]&\geq\E_{C\sim\mc D}\sum_{s:|s|>0}g(C,s)^2\\
    &\geq \E_{C\sim\mc D}\sum_{s:|s|=d+1}g(C,s)^2\\
    &=\frac{(1-\gamma)^{2d}}{4\cdot 15^d}\cdot 2^d\cdot 3^{d-1}=\frac{1}{12}\cdot (1-\gamma)^{2d}\cdot\left(\frac{2}{5}\right)^d.
    \end{aligned}
\end{equation}
This bound can be further improved by considering more paths in Eq.~\eqref{eq:tvdpaths}.
\end{proof}

For completeness, we also summarize known upper bounds for the total variation distance. \cite{Aharonov1996limitations} showed that the KL divergence is upper bounded by $D_{\mathrm{KL}}(\tilde{p}(C)\|U)\leq n\cdot e^{-\gamma d}$. Thus by Pinsker's inequality
\begin{equation}
    \frac{1}{2}\left\|\tilde{p}(C)-U\right\|_1\leq \sqrt{\frac{n}{2}e^{-\gamma d}}.
\end{equation}
Note that the above result holds for any circuit $C$, without averaging. In the anti-concentration regime, \cite{Gao2018efficient} showed an improved upper bound which gives
\begin{equation}
    \E_{C\sim\mc D}\left[\frac{1}{2}\left\|\tilde{p}(C)-U\right\|_1\right]\leq O(1)\cdot e^{-\gamma d}.
\end{equation}

\end{document}